\documentclass[final,12pt,notitlepage]{article}
\usepackage{amsfonts}
\usepackage{amsmath}
\usepackage{amssymb}
\usepackage{mathtools}
\usepackage{amsthm}
\usepackage{subcaption}
\usepackage[svgnames]{xcolor}
\usepackage[
    format        = plain,
    font          = small,
    labelfont     = {bf,it},
    belowskip     = -5pt,
    labelsep      = period
]{caption}
\usepackage{pgfplots}
\usepackage[margin=1.19in]{geometry}
\usepackage[bottom,multiple,splitrule]{footmisc}
\usepackage{graphicx}
\usepackage{natbib}
\usepackage[inline]{enumitem}
\usepackage{hyperref}
\usepackage[onehalfspacing]{setspace}
\usepackage{titlesec}
\usepackage{sectsty}
\usepackage{dsfont}
\usepackage{xargs}
\usepackage[normalem]{ulem}
\usepackage{comment}
\usepackage[capitalize,nameinlink]{cleveref}
\usepackage{bm}
\usepackage[title]{appendix}
\usepackage{tkz-euclide}

\newcommandx{\alitodo}[2][1=]{\todo[linecolor=blue,backgroundcolor=blue!25,bordercolor=blue,#1]{#2}}

\hypersetup{
pdftitle={Optimal Procurement Design: A Reduced-Form Approach},
pdfauthor={Kun Zhang},
pdfkeywords={}
}
\hypersetup{
colorlinks,breaklinks,naturalnames,
citecolor=DarkBlue,urlcolor=Indigo,linkcolor=DarkBlue
}
\usepackage{tikz}
\usetikzlibrary{patterns, decorations.pathreplacing, arrows.meta, angles}
\usetikzlibrary{arrows}
\usetikzlibrary{hobby,snakes}
\usetikzlibrary{shapes.misc, positioning}

\theoremstyle{plain}
\newtheorem{theorem}{Theorem}

\newtheorem{corollary}{Corollary}
\newtheorem{lemma}{Lemma}

\newtheorem{proposition}{Proposition}

\newtheorem{definition}{Definition}

\newtheorem{innercustomthm}{Theorem}

\newenvironment{customthm}[1]
  {\renewcommand\theinnercustomthm{#1}\innercustomthm}
  {\endinnercustomthm}
\theoremstyle{definition}
\newtheorem{example}{Example}
\theoremstyle{remark}

\let\bar\overline

\sectionfont{\color{DarkRed}}
\subsectionfont{\color{DarkRed}}
\subsubsectionfont{\color{DarkRed}}
\begin{document}

\begin{titlepage}
\title{Optimal Procurement Design: \\ A Reduced-Form Approach}
\author{Kun Zhang\thanks{School of Economics, University of Queensland. Email: \href{mailto:kun@kunzhang.org}{
\texttt{kun@kunzhang.org}} \newline
I am indebted to Andreas Kleiner and Alejandro Manelli for their invaluable encouragement and guidance, and I wish to thank the coeditor, Alexander Wolitzky, and four anonymous referees for their excellent editorial comments. For helpful feedback and advice, I am grateful to Nageeb Ali, Hector Chade, Makoto Hanazono, Ravi Jagadeesan, Yixing Li, Simon Loertscher, Dong-Hyuk Kim, Ilia Krasikov, Benny Moldovanu, Juan Ortner, Edward Schlee, Ji Shen, Kunpeng Sun, Alexander Teytelboym, Mark Whitmeyer, Renkun Yang, Mengxi Zhang, and Yimeng Zhang, with special thanks to Claudio Mezzetti. I also thank the seminar audiences at ASU, Auburn, CUFE, UC Riverside, and Xiamen University, as well as the participants at AETW (Auckland). All errors are my own.}}
\date{March 8, 2026}

\maketitle

\bigskip

\begin{abstract}
	Standard procurement models assume that the buyer knows the quality of the good at the time of procurement; however, in many settings, the quality is learned only long after the transaction. We study procurement problems in which the buyer's valuation of the supplied good depends directly on its quality, which is unverifiable and unobservable to the buyer. For a broad class of procurement problems, we identify procurement mechanisms maximizing any weighted average of the buyer's expected payoff and social surplus. The optimal mechanism can be implemented via an auction that restricts sellers to submitting bids within specific intervals.
\end{abstract}

\thispagestyle{empty} 
\vspace{0.4in}

\end{titlepage}

%\newpage

\section{Introduction} \label{ito}
Procurement is essential for many different types of organizations. Unlike typical market purchases, in which goods are often standardized and readily available from many sellers, procurement typically involves acquiring custom goods or services from a limited number of potential suppliers. For example, hospitals rely on procurement to acquire medicines and medical devices \citep{blpt}, governments use it to obtain weapons and public services \citep{lsz}, and firms obtain most of their production inputs as well as many components of their product lines via procurement.

In such procurement contexts, the buyer's valuation of the supplied good or service depends directly on its quality. However, at the time the procurement contract is signed, the quality is often (at least partially) unverifiable in court and unobservable by the buyer. Indeed, although the buyer may have an initial estimate of quality, the true stability of a custom software system or the true reliability of a transportation system, for example, often remains uncertain until long after the contract is executed. As pointed out by \cite{mv95} and \cite{lpv}, in such contexts, buyers who run standard procurement auctions may harbor ``quality concerns''---that is, concerns that aggressively bidding sellers could be offering low-quality goods, potentially leading to undesirable outcomes.

In this paper, we ask the following question: What is the optimal design of the procurement mechanism when quality concerns are present? Under mild regularity conditions, a \textit{bid-restricted auction} (BRA) turns out to be optimal for the buyer. Like a second-price auction, a BRA induces an equilibrium in weakly dominant strategies. The primary difference is that in a BRA sellers are restricted to bidding within a collection of intervals; see \autoref{fig:bra_illu} for an illustration. The seller with the lowest bid wins the auction and supplies the good, and ties are broken uniformly at random. Because bids are only permitted within these intervals, making a bid that lies in a ``gap'' between intervals is infeasible. Accordingly, a seller who would like to submit such a bid ``rounds up'' to the smallest permissible bid above the gap (i.e., the lower endpoint of the next interval) so that the entire excluded range of bids collapses into a single bid. Whenever this bid is the lowest bid submitted, the tie-breaking rule applies, yielding a natural form of pooling (random allocation) over the gap. In most cases, the winning seller receives a payment equal to the second-lowest bid; however, in special circumstances, a ``payment reduction rule'' is applied, resulting in a payment lower than the second-lowest bid. This special rule helps eliminate incentives to bid into the adjacent lower interval to break a tie and thereby increase the probability of winning; as a result, the bid restrictions sustain the intended pooling over the gaps.

More generally, we consider the problem of maximizing any weighted average of the buyer's expected payoff and the social surplus, subject to the requirement that the buyer's expected payoff is non-negative. This requirement can be interpreted as the buyer's individual rationality constraint: the buyer can always decline to procure when her expected net value is negative. It can also be viewed as a reduced-form ``value-for-money'' screen used in procurement appraisal.\footnote{Although quality-related benefits are not literal budget items, many public-sector procurement appraisal frameworks explicitly allow for non-financial performance considerations and other ``unmonetized'' effects to enter value-for-money comparisons; see, for example, Australian Government, \href{https://www.finance.gov.au/sites/default/files/2025-10/Commonwealth-Procurement-Rules-2025.pdf}{\textit{Commonwealth Procurement Rules 2025}}; HM Treasury, \href{https://assets.publishing.service.gov.uk/media/62443d2c8fa8f5277b365ad7/Green_Book_supplementary_guidance_-_Value_for_Money.pdf}{\textit{Green Book Supplementary Guidance: Value for Money}}; and U.S. OMB, \href{https://www.whitehouse.gov/wp-content/uploads/2023/11/CircularA-94.pdf}{\textit{Circular A-94}}.}
In this more general setting, an optimal procurement mechanism is still a BRA, although a random reserve price might need to be introduced.

\begin{figure}
    \centering
    \begin{tikzpicture}[scale = 7]

\def\blevel{0};

\draw (0,0) -- (1,0) node[right] {};
\node[below] at (0, -0.6pt) {\(\underline{b}_1 = 0\)};
\node[below] at (0.4,-0.5pt) {\(\overline{b}_1\)};
\node[below] at (0.7,-0.7pt) {\(\underline{b}_2\)};
\node[below] at (1, -0.5pt) {\(\overline{b}_2 = 1\)};

\draw[DarkBlue, line width = 1.50mm] (0,\blevel) -- (0.4,\blevel);

\draw[DarkBlue, line width = 1.20mm, shift={(0,\blevel)}] (0pt,0.49pt) -- (0pt,-0.49pt);

\draw[DarkBlue, line width = 1.20mm, shift={(0.4,\blevel)}] (0pt,0.49pt) -- (0pt,-0.49pt);

\draw[DarkBlue, line width = 1.50mm] (0.7,\blevel) -- (1,\blevel);

\draw[DarkBlue, line width = 1.20mm, shift={(0.7,\blevel)}] (0pt,0.49pt) -- (0pt,-0.49pt);

\draw[DarkBlue, line width = 1.20mm, shift={(1,\blevel)}] (0pt,0.49pt) -- (0pt,-0.49pt);
    
\end{tikzpicture}
    \caption{Illustration of a BRA. Sellers who wish to participate in the auction must submit a bid \(b\) within the two bid intervals, \([\underline{b}_1, \overline{b}_1]\) and \([\underline{b}_2, \overline{b}_2]\). In other words, bids strictly between \(\overline{b}_1\) and \(\underline{b}_2\) are not allowed.}
    \label{fig:bra_illu}
\end{figure}

The key insight of this result is that the optimal BRA balances the tension between inducing price competition and alleviating quality concerns. It does so by fostering price competition within bid intervals, where quality concerns are mild; it simultaneously prohibits bids in the gaps between intervals, with the aim of dampening incentives for aggressive price competition in the range in which such competition is most problematic from a quality perspective. Restricting bids to suitably chosen intervals reconciles quality concerns with the simplicity of a second-price-like auction format. Put differently, while ``global'' quality concerns might necessitate abandoning competitive bidding altogether, ``local'' quality concerns can effectively be addressed by excluding only those bids that would otherwise fuel undesirable price competition.

We establish the optimality of BRA using a reduced-form approach. Specifically, we transform the procurement design problem into the problem of choosing an interim allocation rule, also known as a ``reduced form'' in the literature, that maximizes a linear functional identified by some virtual surplus. An interim allocation rule specifies the expected probability that the buyer procures a good from a particular seller based on her report of her private information. The interim allocation rule must satisfy a monotonicity constraint to ensure truth-telling and \citeauthor{border1991}'s (\citeyear{border1991}) condition to guarantee feasibility. We then apply techniques from linear optimization under a majorization constraint, developed by \cite{kms}, to solve for the optimal interim allocation. We complete the argument by verifying that a BRA induces the optimal interim allocation.

The study of procurement problems with unverifiable and unobservable quality was pioneered by \cite{mv95}.\footnote{See \cite{che2008} for a broader survey of the procurement literature; for a summary of recent papers on the performance of certain ad hoc mechanisms in the presence of quality concerns, see \citet[p.~1507]{lpv}.} They point out that in many cases, a standard procurement auction may perform poorly due to quality concerns; instead, it may be optimal to sequentially render take-it-or-leave-it offers to potential sellers. \cite{mv04} show that some ``hybrid mechanisms,'' which combine elements of sequential offers and auctions, can be optimal in specific procurement settings. In contrast to these studies, our work identifies procurement mechanisms that maximize any weighted average of the buyer's expected payoff and social surplus across a broad class of procurement problems.

\cite{lpv} show that when the buyer's virtual surplus is single-peaked, a mechanism called the lowball lottery auction (LoLA) is optimal. LoLA differs from a standard second-price procurement auction with a reserve price only in that sellers are not allowed to bid below a certain ``floor price.'' 
Our work generalizes \cite{lpv} by identifying the optimal procurement mechanisms in a broader class of environments;\footnote{This is not merely a technical curiosity; violations of their assumption may naturally arise due to the nature of the procurement setting (see \autoref{example:reli} for a concrete example).} indeed, when the buyer's virtual surplus is single-peaked, the optimal BRA is a LoLA. More importantly, our results reveal that what matters is not merely whether a floor price is used but whether sellers are disallowed from submitting bids in certain intervals in which quality concerns are severe. Furthermore, \cite{lpv} do not impose the restriction that the buyer's expected payoff must be non-negative and therefore do not offer insights into how this constraint should be addressed, an issue that can be relevant in practice.

In a BRA, the gaps between bid intervals can be interpreted as ``pooling regions.'' Other papers in the procurement literature also feature pooling regions as a part of their optimal mechanisms. \cite{burguet2012limited} and \cite{chillemi2014optimal} study optimal procurement mechanisms under the risk of default or contract breach. They show, among other findings, that pooling can be useful in reducing the probability of default. \cite{cck} study procurement design when sellers may collude without using side payments. They show that a collusion-proof mechanism requires pooling on intervals on which the concavified type distribution is affine.\footnote{In the context of college admission contests, \cite{krishna2026pareto} determine the intervals on which pooling introduces Pareto improvement relative to separation by concavifying the type distribution.} In this paper, by contrast, pooling regions are introduced to mitigate quality concerns, and they are identified by concavifying certain virtual surplus functions.

The remainder of this paper is organized as follows. \autoref{mpp} sets up and transforms the optimal procurement design problem. In \autoref{section:buyer-optimal} we study the special case of maximizing the buyer's expected payoff; \autoref{sop} discusses the general problem of maximizing any weighted average of the buyer's payoff and the social surplus, with the constraint that the buyer's payoff is bounded below by zero. \autoref{s:conclusion} concludes.

\section{The procurement problem} \label{mpp}
The model, which is essentially the same as that in \citet{mv95, mv04}, consists of one buyer and \(n > 1\) symmetric potential sellers. The buyer would like to procure one unit of a product from one of the potential sellers. Each seller \(s\) has private information, \(q_s \in [0,1]\). We refer to \(q_s\) as the quality of the product offered by seller \(s\); it can also be interpreted as seller \(s\)'s cost or reservation value.\footnote{These interpretations are discussed further in \autoref{dnp}.} Qualities are independently and identically distributed according to a continuous density function \(f\); we denote the corresponding cumulative distribution function by \(F\). We also assume that \(f(q)>0\) for \(q \in (0,1]\) and that \(f(0) = 0\) only if
\(\lim_{q \to 0} \left(F(q)/f(q)\right) = 0.\)
Since we assume that the potential sellers are symmetric, we often suppress the subscript of a seller's quality. 

All agents in our model are expected utility maximizers. If the buyer procures a good from a seller and a transfer \(t\) is made, the seller's payoff is \(t-q\). If a seller does not sell, her payoff is zero. The buyer's valuation for a good of quality \(q\) is a continuous function \(v(q)\); we assume that \(v(0) \ge 0\). If the buyer makes a transfer \(t\) and receives an object of quality \(q\), her payoff is \(v(q)-t\); if no trade occurs, the buyer's payoff is zero.

By the revelation principle, it suffices to focus on direct mechanisms. A direct mechanism is characterized by a pair of functions, \(p_s: [0,1]^n \to [0,1]\) and \(t_s: [0,1]^n \to \mathbb{R}\), for each seller \(s\). If the sellers report \(\boldsymbol{q} \coloneqq (q_1, \ldots, q_n)\), the buyer procures from seller \(s\) with probability \(p_s(\boldsymbol{q})\) and she makes transfer \(t_s(\boldsymbol{q})\) to seller \(s\); we call \(p_s(\cdot)\) the \textbf{ex-post allocation probability} for seller \(s\). Because the buyer wishes to acquire (at most) one unit of the product, for each \(\boldsymbol{q} \in [0,1]^n\), the following feasibility constraint must hold:
\begin{equation} \label{eq:feasibility}
\sum_{s=1}^n p_s(\boldsymbol{q}) \le 1. \tag{F}
\end{equation}
Condition \eqref{eq:feasibility} requires that the probability that the buyer buys from one of the potential sellers is less than or equal to 1.

If seller \(s\) reports \(q'_s\) and assumes that the rest of the sellers report truthfully, she would expect that the buyer procures from her with probability
\[
P_s(q'_s) \coloneqq \int p_{s}\left(q'_{s}, \boldsymbol{q}_{-s}\right) f^{n-1}\left(\boldsymbol{q}_{-s}\right) \, \mathrm{d} \boldsymbol{q}_{-s},
\]
where \(\boldsymbol{q}_{-s} \coloneqq (q_1, \ldots, q_{s-1}, q_{s+1}, \ldots, q_n)\) and \(f^{n-1}(\boldsymbol{q}_{-s}) \coloneqq \prod_{k \ne s} f(q_k)\); she would expect to receive a monetary transfer of
\[
T_s(q'_s) \coloneqq \int t_{s}\left(q'_{s}, \boldsymbol{q}_{-s}\right) f^{n-1}\left(\boldsymbol{q}_{-s}\right) \, \mathrm{d} \boldsymbol{q}_{-s}.
\]
We call \(P_s(\cdot)\) the \textbf{interim allocation probability} for seller \(s\). 
Then the expected payoff of seller \(s\) with quality \(q_s\) from reporting \(q'_s\) can be written as
\[\pi_s(q'_s \mid q_s) \coloneqq T_s(q'_s) - q_s P_s(q'_s),\]
and we let \(\pi_s(q_s)\coloneqq\pi_s(q_s \, | \, q_s)\). For a given direct mechanism \(\{p_s,t_s\}_{s = 1}^n\) , the (expected) social surplus is 
\[\sum_{s=1}^{n} \int_{[0,1]^{n}}\left[v\left(q_{s}\right) - q_s\right] p_{s}(\boldsymbol{q}) f^n(\boldsymbol{q}) \, \mathrm{d}\boldsymbol{q},\]
where \(f^n(\boldsymbol{q}) \coloneqq \prod_{s=1}^n f(q_s)\); the buyer's expected payoff is 
\begin{align}
    \pi_{b} \coloneqq & \sum_{s=1}^{n} \int_{[0,1]^{n}}\left[v\left(q_{s}\right) p_{s}(\boldsymbol{q})-t_{s}(\boldsymbol{q})\right] f^n(\boldsymbol{q}) \, \mathrm{d}\boldsymbol{q}. \label{bbp}
\end{align}
We say that a direct mechanism \(\{p_s,t_s\}_{s = 1}^n\) is \textbf{incentive-compatible} if every seller truthfully reports her quality; formally, for each seller \(s\), all \(q'_s \in [0,1]\) and (almost) all \(q_s \in [0,1]\), 
\(\pi_s(q_s) \ge \pi_s(q'_s \mid q_s)\). We
say that it is \textbf{individually rational} if the buyer and all sellers are willing to participate in the mechanism; that is, \(\pi_b \ge 0\) and \(\pi_s(q_s) \ge 0\) for each seller \(s\) and \(q_s \in [0,1]\). 

\autoref{lemma:icir} characterizes the set of incentive-compatible direct mechanisms and eliminates transfers from the buyer's expected payoff. The proof is standard and hence omitted.

\begin{lemma} \label{lemma:icir}
    Let \(\{p_s(\cdot)\}_{s=1}^{n}\) be a collection of ex-post allocation probabilities, where \(p_s: [0,1]^n \to [0,1]\), satisfying \eqref{eq:feasibility}.
    \begin{itemize}
        \item[(1)] There exists a collection of transfers \(\{t_s(\cdot)\}_{s=1}^{n}\) such that \(\{(p_s(\cdot), t_s(\cdot))\}_{s=1}^n\) is incentive-compatible if and only if for each \(s = 1, \ldots, n\), \(P_s(\cdot)\) is decreasing.\footnote{We use ``decreasing'' in the weak sense: ``strict'' will be added whenever needed.}
        
        \item[(2)] For any incentive-compatible direct mechanism \(\{(p_s(\cdot), t_s(\cdot))\}_{s=1}^n\),
        \begin{itemize}
            \item[i.] it is individually rational if and only if \(\pi_b \ge 0\) and \(\pi_s(1) \ge 0\) for each \(s=1, \ldots, n\); and
            
            \item[ii.] the buyer's expected payoff is given by
            \[\pi_{b} = \sum_{s=1}^{n} \int_{[0,1]^n}\left[v\left(q_{s}\right) - q_s - \frac{F\left(q_{s}\right)}{f\left(q_{s}\right)}\right] p_{s}(\boldsymbol{q}) f^n(\boldsymbol{q}) \, \mathrm{d} \boldsymbol{q} - \sum_{s=1}^{n} \pi_{s}(1).\]
        \end{itemize}
    \end{itemize}
\end{lemma}

By \autoref{lemma:icir}, the weighted average of the buyer's expected payoff (with weight \(\gamma\)) and the social surplus (with weight \(1-\gamma\)) can be written as
\begin{align*}
     & \gamma \left\{\sum_{s=1}^{n} \int_{[0,1]^n}\left[v\left(q_{s}\right) - q_s - \frac{F\left(q_{s}\right)}{f\left(q_{s}\right)}\right] p_{s}(\boldsymbol{q}) f^n(\boldsymbol{q}) \, \mathrm{d} \boldsymbol{q} - \sum_{s=1}^{n} \pi_{s}(1)\right\} + \\
    & (1-\gamma) \sum_{s=1}^{n} \int_{[0,1]^n}\left[v\left(q_{s}\right) - q_s\right] p_{s}(\boldsymbol{q}) f^n(\boldsymbol{q}) \, \mathrm{d} \boldsymbol{q} \nonumber \\
     = \, &  \sum_{s=1}^{n} \int_{[0,1]^n}\left[v\left(q_{s}\right) - q_s - \gamma \frac{F(q_s)}{f(q_s)}\right] p_{s}(\boldsymbol{q}) f^n(\boldsymbol{q}) \, \mathrm{d} \boldsymbol{q} - \gamma \sum_{s=1}^{n} \pi_{s}(1).
\end{align*}
Since individual rationality for the sellers is equivalent to \(\pi_s(1) \ge 0\) for each \(s=1, \ldots, n\), to maximize this weighted average, we set \(\pi_s(1) = 0\) for all \(s = 1, \ldots, n\).\footnote{For \(\gamma > 0\), the condition \(\pi_s(1) = 0\) for all \(s = 1, \ldots, n\) is necessary for a mechanism to be optimal. For \(\gamma = 0\), setting \(\pi_s(1) = 0\) does not change the social surplus but does relax the buyer's individual rationality constraint.} Therefore, we consider the following maximization problem:
\begin{align}
    \max_{\{p_s\}_{s=1}^n}~~~ &~~ \sum_{s=1}^{n} \int_{[0,1]^n}\left[v\left(q_{s}\right) - q_s - \gamma \frac{F(q_s)}{f(q_s)}\right] p_{s}(\boldsymbol{q}) f^n(\boldsymbol{q}) \, \mathrm{d} \boldsymbol{q} \label{eq:obj_ori} \\
    \text{subject to} &~~ \eqref{eq:feasibility} \nonumber \\
    &~~ P_s(\cdot) \text{ is decreasing for each } s = 1, \ldots, n \nonumber \\
    & ~~ \sum_{s=1}^{n} \int_{[0,1]^n}\left[v\left(q_{s}\right) - q_s - \frac{F\left(q_{s}\right)}{f\left(q_{s}\right)}\right] p_{s}(\boldsymbol{q}) f^n(\boldsymbol{q}) \, \mathrm{d} \boldsymbol{q} \ge 0. \nonumber
\end{align}
The inequality constraint requires that the buyer's expected payoff is non-negative.

When \(\gamma = 1\), the above problem becomes the buyer's expected payoff maximization problem, which can be written as
\begin{align*}
    \max_{\{p_s\}_{s=1}^n}~~~ &~~ \sum_{s=1}^{n} \int_{[0,1]^n}\left[v\left(q_{s}\right) - q_s - \frac{F\left(q_{s}\right)}{f\left(q_{s}\right)}\right] p_{s}(\boldsymbol{q}) f^n(\boldsymbol{q}) \, \mathrm{d} \boldsymbol{q} \\
    \text{subject to} &~~ \eqref{eq:feasibility} \\
    &~~ P_s(\cdot) \text{ is decreasing for each } s = 1, \ldots, n.
\end{align*}
The inequality constraint is not needed because setting \(p_s(\boldsymbol{q}) = 0\) for all \(\boldsymbol{q} \in [0,1]^n\) and for \(s = 1, \ldots, n\) is feasible, ensuring that the value of the problem is non-negative.

\subsection{Transforming the problem: A reduced-form approach} \label{rfa}
Instead of solving directly for optimal ex-post allocation probabilities \citep[e.g.,][]{myerson1981,mv95}, we take a reduced-form approach. We first solve for the optimal interim allocation probabilities \(\{P_s(\cdot)\}_{s=1}^{n}\) and then identify a trading mechanism that is consistent with them. This approach is valid because a seller's report affects her incentives only through the interim allocation probability it generates, and the feasibility constraint can be expressed in terms of interim allocation probabilities as well. Working with interim allocations is also convenient: each \(P_s(\cdot)\) is a function of a single variable, whereas ex-post allocation probabilities are functions of \(n\) variables.

To this end, note that the objective function \eqref{eq:obj_ori}, in terms of interim allocation probabilities, can be written as
\[
\sum_{s=1}^{n} \int_{0}^{1}\left[v\left(q_{s}\right) - q_s - \gamma\frac{F\left(q_{s}\right)}{f\left(q_{s}\right)}\right] P_{s}(q_s) f(q_s) \, \mathrm{d} q_s.
\]
Say that the collection of interim allocation probabilities \(\{P_s\}_{s=1}^n\), where \(P_s: [0,1] \to [0,1]\) for each \(s\), is \textbf{implementable} if there exists a collection of ex-post allocation probabilities \(\{p_s\}_{s=1}^n\) satisfying \eqref{eq:feasibility} that induces \(\{P_s\}_{s=1}^n\) as its interim allocations; that is, for each \(s = 1, \ldots, n\) and all \(q_s \in [0,1]\),
\[
P_s(q_s) = \int p_{s}\left(q_{s}, \boldsymbol{q}_{-s}\right) f^{n-1}\left(\boldsymbol{q}_{-s}\right) \, \mathrm{d} \boldsymbol{q}_{-s}.
\]
Since the sellers are symmetric, it is without loss to restrict attention to symmetric interim allocations; thus, we can drop the subscript \(s\) from \(P_s\) and \(q_s\) and write \(P\) and \(q\) instead. Consequently, the objective further reduces to
\begin{equation} \label{eq:obj_red}
n \int_{0}^{1} \left[v(q) - q - \gamma \frac{F(q)}{f(q)}\right] P(q) f(q) \, \mathrm{d}q,
\end{equation}
and the buyer's expected payoff can be written as
\[n \int_{0}^{1} \left[v(q) - q - \frac{F(q)}{f(q)}\right] P(q) f(q) \, \mathrm{d}q.\]

As will become evident, introducing the quantile \(s = F(q)\) makes identifying an optimal interim allocation straightforward. Define
\(\widetilde{P}(s) \coloneqq P(F^{-1}(s))\)
as the \textbf{quantile interim allocation}; because \(f > 0\) on \((0,1]\), by \autoref{lemma:icir}, a collection of ex-post allocation probabilities is a part of an incentive-compatible mechanism if and only if the induced quantile interim allocation is decreasing. 

Let \(\mathcal{P}\) denote the set of decreasing functions mapping \([0,1]\) into itself. For any \(\widetilde{P}', \widetilde{P}'' \in \mathcal{P}\), we say that \(\widetilde{P}'\) \textbf{majorizes} \(\widetilde{P}''\), denoted by \(\widetilde{P}'' \prec \widetilde{P}'\), if 
\begin{equation}\label{equation:majorization}
    \int_{0}^{x} \widetilde{P}''(s) \, \mathrm{d} s \leq \int_{0}^{x} \widetilde{P}'(s) \,\mathrm{d}s ~~\text { for all } x \in[0,1],
\end{equation}
with equality for \(x=1\). Say that \(\widetilde{P}'\) \textbf{weakly majorizes} \(\widetilde{P}''\), denoted by \(\widetilde{P}'' \prec_{w} \widetilde{P}'\), if \eqref{equation:majorization} holds but equality at \(x=1\) is not required. Intuitively, \(\widetilde{P}'\) majorizes \(\widetilde{P}''\) means that \(\widetilde{P}'\) is more variable, or more ``spread out,'' than \(\widetilde{P}''\). In our context, this variability corresponds to the degree of discrimination among seller qualities: a trading mechanism that more sharply distinguishes between qualities generates a more variable interim allocation than one that pools qualities. While majorization keeps the ex-ante allocation probability \(\int_{0}^{1}\widetilde{P}(s)\,\mathrm{d}s\) intact, weak majorization allows the ex-ante allocation probability to be strictly lower.

\citeauthor{border1991}'s (\citeyear{border1991}) celebrated theorem characterizes the set of implementable interim allocations. \autoref{boc} translates Border's condition into our terminology.\footnote{To the best of our knowledge, \cite{ggkms} is the first paper that connects \citeauthor{border1991}'s condition to majorization (see their footnote 4). We omit the proof of \autoref{boc} since it can be proved by slightly modifying the proof of, for example, Theorem 1 in \cite{hr15} or Theorem 3 in \cite{kms}.} Let \(\widetilde{P}^*(s) \coloneqq (1-s)^{n-1}\); it is not difficult to see that \(\widetilde{P}^*(\cdot)\) is the quantile interim allocation induced by the allocation that always procures from the seller with the lowest quality.

\begin{lemma}[Border's condition] \label{boc}
    A decreasing interim allocation probability \(P\) is implementable if and only if the associated quantile interim allocation \(\widetilde{P}(s)\) is weakly majorized by \(\widetilde{P}^*\).
\end{lemma}

Border's condition can be interpreted as follows: discriminating less by pooling qualities is always feasible, but discriminating more than the rule that always trades with the seller with the lowest quality is infeasible. Moreover, weak majorization allows for a reduction in the probability of trade through exclusion.\footnote{Since \(\int_{0}^{1}\widetilde{P}^{*}(s)\,\mathrm{d}s = 1/n\), equality at \(x=1\) in \(\widetilde{P}\prec \widetilde{P}^{*}\) implies \(n\int_{0}^{1}\widetilde{P}(s)\,\mathrm{d}s = 1\), i.e., the buyer procures with probability one, whereas weak majorization allows this probability to be strictly below one.}

To simplify notation, let 
\[h_{\gamma}(q)\coloneqq v(q) - q - \gamma \frac{F(q)}{f(q)}\] 
denote the integrand of \eqref{eq:obj_red} when the weight on the buyer's expected payoff is \(\gamma\); we call \(h_{\gamma}\) the \textbf{weighted virtual surplus}. In particular, let
\[
g(q) \coloneqq h_1(q) = v(q) - q - \frac{F(q)}{f(q)}
\]
denote the buyer's virtual surplus, which is the buyer's valuation net of the virtual cost (i.e., actual cost plus information rent). By \autoref{boc}, the quantile interim allocation of a direct mechanism that maximizes the weighted average can be found by solving\footnote{A solution to problem \eqref{problem:weighted} exists: \(\Omega_w(\widetilde{P}^*)\) is compact (in the \(L^1\) norm topology) by Helly's selection theorem \citep[see, for example, Proposition 1 in][]{kms}, which implies that the constraint set
\[
\Omega_w(\widetilde{P}^*) \cap 
\left\{\widetilde{P} \in \mathcal{P}: \int_{0}^{1} g(F^{-1}(s)) \widetilde{P}(s) \, \mathrm{d}s \ge 0 \right\}
\]
is the intersection of a closed set and a compact set and hence also compact. The argument is completed by noting that the objective function is continuous.}
\begin{align}
    \max_{\widetilde{P} \in \Omega_w(\widetilde{P}^*)}~~ & \int_{0}^{1} h_\gamma(F^{-1}(s)) \widetilde{P}(s) \, \mathrm{d}s \label{problem:weighted} \\
    \text{s.t.}~~~~~~ & \int_{0}^{1} g(F^{-1}(s)) \widetilde{P}(s) \, \mathrm{d}s \ge 0, \nonumber
\end{align}
where
\[
\Omega_w(\widetilde{P}^*) \coloneqq \left\{\widetilde{P} \in \mathcal{P} : \widetilde{P} \prec_w \widetilde{P}^*\right\}.
\]
In particular, the buyer-optimal interim allocation should solve the following problem:
\begin{align}
    \max_{\widetilde{P} \in \Omega_w(\widetilde{P}^*)}~~ & \int_{0}^{1} g(F^{-1}(s)) \widetilde{P}(s) \, \mathrm{d}s. \label{problem:buyer}
\end{align}

\section{Buyer-optimal procurement mechanisms} \label{section:buyer-optimal}
For cleaner intuition and simpler notation, in this section we study a special case of our problem: the buyer's expected payoff maximization problem, in which the welfare weight on the buyer is \(\gamma = 1\). In \autoref{boi} we solve problem \eqref{problem:buyer}, and in \autoref{bim} we identify a trading mechanism that implements the solution we find. All proofs in this section are relegated to \autoref{buyer_optimal_proofs}.

\subsection{Buyer-optimal interim allocation} \label{boi}
To ensure that the monotonicity constraint in problem \eqref{problem:buyer} holds, the ironing technique (\citeauthor{myerson1981}, \citeyear{myerson1981}; \citeauthor{toikka2011}, \citeyear{toikka2011}) may be required. Define \(\widetilde{g}(s) \coloneqq g(F^{-1}(s))\) as the (buyer's) \textbf{quantile virtual surplus}, and let
\[
G(s) \coloneqq \int_{0}^{s} \widetilde{g}(x) \, \mathrm{d}x.
\]
Let \(\overline{G}\) be the concave hull of \(G\); that is, \(\overline{G}(x) \coloneqq \sup \{y: (x,y) \in \mathrm{co}(G)\}\), where \(\mathrm{co}(G)\) is the convex hull of the hypograph of \(G\). Equivalently, \(\overline{G}\) is the pointwise smallest concave function that lies above \(G\). Call \(\overline{g} \coloneqq \overline{G}'\) the \textbf{ironed quantile virtual surplus};\footnote{Because \(\overline{G}\) is concave, it is differentiable almost everywhere. At points where it is not differentiable, we define \(\overline{g}\) as the right derivative.} since \(\overline{G}\) is concave, \(\overline{g}\) is decreasing.

\autoref{p:oia} characterizes a buyer-optimal interim allocation. Adopting the convention that \(\sup \varnothing = 0\), define
\begin{equation} \label{ivs}
    \overline{S} \coloneqq \sup \{s \in [0,1] : \overline{g}(s) \ge 0\};
\end{equation}
in words, \(\overline{S}\) is the highest quantile at which the ironed quantile virtual surplus is non-negative. Intuitively, \(\overline{S}\) is the optimal exclusion cutoff: procuring from sellers with \(s > \overline{S}\) is too expensive relative to the value they provide. Our assumptions on \(v\) and \(F\) guarantee that \(\overline{S}\) is well-defined.

\begin{proposition} \label{p:oia}
    Let \(\{\left[\underline{s}_{i}, \bar{s}_{i}\right)\}_{i \in \mathcal{I}}\) denote a countable collection of disjoint intervals with \(\left[\underline{s}_{i}, \bar{s}_{i}\right) \subseteq [0,\overline{S}]\) for each \(i \in \mathcal{I}\), such that
    \begin{itemize}
        \item \(\overline{G}\) is affine on \(\left[\underline{s}_{i}, \bar{s}_{i}\right)\) for each \(i \in \mathcal{I}\); and
        \item \(\overline{G} = G\) on \([0,\overline{S}] \big\backslash \bigcup_{i \in \mathcal{I}}\left[\underline{s}_{i}, \overline{s}_{i}\right)\).
    \end{itemize}
    Then an optimal interim allocation \(\hat{P}\) satisfies
    \begin{equation} \label{eq:oip}
    	\hat{P}(s) = 
    \begin{cases}
    (1-s)^{n-1} & \text{ if } s \in [0,\overline{S}] \big\backslash \bigcup_{i \in \mathcal{I}}\left[\underline{s}_{i}, \overline{s}_{i}\right), \\
    \frac{\int_{\underline{s}_{i}}^{\overline{s}_{i}} (1-s)^{n-1} \, \mathrm{d}s}{\overline{s}_{i} - \underline{s}_{i}} & \text{ if } s \in \left[\underline{s}_{i}, \overline{s}_{i}\right), \\
    0 & \text{ if } s \in (\overline{S},1].
    \end{cases}
    \end{equation}
\end{proposition}

    \autoref{p:oia} follows from the results concerning maximizing linear functionals under a majorization constraint in \cite{kms}. 
    To better understand the intuition, consider first the simplest case in which the buyer's virtual surplus \(g\) is decreasing (and therefore so is the quantile virtual surplus \(\widetilde{g}\)). This implies that the buyer's quality concerns are globally mild: roughly speaking, the buyer's marginal valuation of quality is dominated by the marginal cost of quality everywhere on the unit interval. In this case, \(G\) is concave, implying \(G = \overline{G}\) and \(\widetilde{g} = \overline{g}\); thus, ironing is not needed. For \(s > \overline{S}\), \(\widetilde{g}(s)\) is strictly negative, and thus the buyer would like to choose the minimal interim allocation 0 on \((\overline{S},1]\). For \(s \le \overline{S}\), the buyer would like to have a higher (lower) interim allocation when \(\overline{g}\) is larger (smaller). Since \(\overline{g}\) is decreasing, she would like to set \(\widetilde{P}(s)\) as high (low) as possible when \(s\) is small (large). This favors the most variable interim allocation \(\widetilde{P}^* = (1-s)^{n-1}\), which corresponds to the buyer always procuring from the seller with the lowest quality.

    When the buyer's virtual surplus is not decreasing, \(\widetilde{g}\) must be increasing around some quantile. This implies that quality concerns are locally severe: the value of procuring higher quality rises faster than the cost. Ideally, the buyer would prefer to favor higher-quality sellers in this region, but incentive compatibility requires the interim allocation probability to decrease with quality. Facing this conflict, the best the buyer can do is to make the interim allocation constant on an interval around that quantile. To satisfy feasibility (i.e., \(\widetilde{P} \prec_w \widetilde{P}^*\)), it suffices to set the interim allocation to the conditional mean of \(\widetilde{P}^*\) on that interval. The ironing procedure identifies these intervals, which are the maximal intervals on which \(G\) differs from \(\overline{G}\); outside of these intervals, the logic in the previous paragraph applies, establishing the optimality of \(\hat{P}\) as defined in \eqref{eq:oip}. This is illustrated in \autoref{fig:prop1_example}.

\begin{figure}[ht]%
    \begin{minipage}{.47\linewidth} 
        \centering
        \subfloat[]{
        \begin{tikzpicture}[xscale=5,yscale=11,>=Stealth]
            % Axes
        	\draw[thick, ->] (0,0) node[below left] {\small \(O\)} -- (1.15,0) node[below] {\small \(q\)};
        	\draw[thick, ->] (0,0) -- (0,0.42);

        	\draw[NavyBlue, thick, domain=0:1] plot (\x, {pow(\x,2) - (2/3)*pow(\x,3)});

            \draw[orange, thick] (0,0) -- (0.75,0.28125);
            \draw[orange, thick, domain=0.75:1] plot (\x, {pow(\x,2) - (2/3)*pow(\x,3)});

            % Labels
        	\node[below] at (1,0) {\small 1};
        	\node[right, NavyBlue] at (0.52,0.17) {\small \(G(q)\)};
        	\node[above, orange] at (0.30,0.13) {\small \(\overline{G}(q)\)};

            % Cutoff marker q* = 3/4
        	\draw[dotted, thick] (0.75,0.28125) -- (0.75,0) node[below]{\small 0.75};
        \end{tikzpicture}
        \label{fig:prop1G}}
    \end{minipage}
    \begin{minipage}{.47\linewidth}
        \centering
        \subfloat[]{
        \begin{tikzpicture}[xscale=5,yscale=4,>=Stealth]
            % Axes
        	\draw[thick, ->] (0,0) node[below left] {\small \(O\)} -- (1.15,0) node[below] {\small \(q\)};
        	\draw[thick, ->] (0,0) -- (0,1.16);

            % P^*(q) = 1-q for n=2
        	\draw[NavyBlue, thick, domain=0:1] plot (\x, {1-\x});
        	\node[above right, NavyBlue] at (0.1,0.9) {\small \(P^*(q)=1-q\)};

            % \hat P(q): pooled on [0,3/4] at 5/8, then equals 1-q on (3/4,1]
        	\node[left] at (0,0.625) {\small \(5/8\)};
        	\draw[red, thick] (0,0.625) -- (0.75,0.625);
        	\draw[red, thick] (0.75,0.25) node[above right]{\small \(\hat{P}(q)\)} -- (1,0);

            % Cutoff marker q* = 3/4
        	\draw[dotted, thick] (0.75,0.625) -- (0.75,0) node[below]{\small 0.75};

            % End labels
        	\node[below] at (1,0) {\small 1};
        	\node[left] at (0,1) {\small 1};
    	\end{tikzpicture}
        \label{fig:prop1P}}
    \end{minipage}
 \caption{\small Illustration of \autoref{p:oia} for \(v(q)=-2q^2+4q\), \(q\sim U[0,1]\) and \(n=2\). In panel (a), the blue curve is \(G\) and the orange curve is its concave hull, \(\overline{G}\). In panel (b), the blue curve is \(P^*(q)=1-q\) and the red curve is the optimal interim allocation, \(\hat P(q)\).}%
 \label{fig:prop1_example}%
\end{figure}

\subsection{Implementation: Bid-restricted auctions} \label{bim}
To find the buyer-optimal procurement mechanism, we first need to find some trading mechanisms that implement the optimal interim allocation rule \(\hat{P}(s)\) identified in \autoref{p:oia}. To this end, we introduce the following class of mechanisms.

\begin{definition} \label{def:bra}
    A \textbf{bid-restricted auction (BRA)} is a sealed-bid auction with \(M \in \mathbb{N}\) \textbf{bid intervals} \(\{[\underline{b}_i, \overline{b}_i]\}_{i=1}^{M}\), where \(\underline{b}_1 \ge 0\), \(\overline{b}_M \le 1\), and for all \(i=1,\ldots, M\), \(\underline{b}_i \le \overline{b}_i\) and \(\overline{b}_i < \underline{b}_{i+1}\), with the following rules:
    \begin{itemize}
        \item Any seller who wishes to participate must submit a bid \(b\) within one of the bid intervals; i.e., \(b \in \cup_{i=1}^{M}[\underline{b}_i, \overline{b}_i]\).
        \item The seller whose bid is the lowest wins the auction; ties are broken with equal probability among the sellers submitting the lowest bids.
        \item If the winning bid is the only bid in its bid interval \([\underline{b}_i, \overline{b}_i]\) for some \(i = 1, \ldots, M\), and if the second-lowest bid equals \(\underline{b}_j\) for some \(j > i\) with a total of \(k\) sellers who bid \(\underline{b}_j\),\footnote{If only one seller submits a bid, we take the second-lowest bid to be \(\overline{b}_M\).} then the winning seller receives a payment of \((\underline{b}_j + k\overline{b}_{j-1})/(k+1)\).
        \item Otherwise, the winning seller receives a payment equal to the second-lowest bid.
    \end{itemize}
\end{definition}

Intuitively, a BRA is similar to a second-price auction, with two key exceptions: in a BRA, (i) sellers are allowed to bid only within certain bid intervals; and (ii) a payment reduction rule applies in specific cases. Under this rule, if the winning bid is the only bid in such a bid interval, the winning seller receives a payment that is determined by the other sellers' bids, which may be lower than the second-lowest bid. The following result establishes an important property of BRAs; to simplify the statement, we adopt the convention that \(\overline{b}_0 := 0\). 

\begin{lemma} \label{l:wd}
    In a BRA, it is a weakly dominant strategy for any seller to not bid if her quality exceeds \(\overline{b}_M\), bid her quality when \(q \in [\underline{b}_i, \overline{b}_i)\) for some \(i=1, \ldots, M\), and bid \(\underline{b}_{i}\) if \(q \in [\overline{b}_{i-1}, \underline{b}_{i})\) for some \(i=1, \ldots, M\).
\end{lemma}

The structure inherited from a second-price auction and the payment reduction rule together imply that bidding as described in \autoref{l:wd} is a weakly dominant strategy. As in standard second-price auctions, in BRAs a seller's bid affects the payment she receives only indirectly, through its influence on the identity of the winner. In particular, except in a narrow set of contingencies, the winner is paid the second-lowest bid. Consequently, the standard second-price auction logic implies that if a seller's quality lies in a bid interval \([\underline b_i,\overline b_i)\), bidding her quality is weakly dominant. Likewise, if her quality lies in a gap \([\overline b_{i-1},\,\underline b_i)\), bidding above \(\underline b_i\) can only reduce the probability of winning and cannot be profitable.

The only departure from the standard second-price payment is driven by the ``gaps'' between the bid intervals. To see this, %fix \(i\) and 
consider a seller with quality \(q\in[\overline b_{i-1},\,\underline b_i)\). Under the strategy described in \autoref{l:wd}, the seller bids \(\underline b_i\), and the same is true for any seller whose quality lies in \([\overline b_{i-1},\,\underline b_i)\); thus, ties at \(\underline b_i\) may occur with positive probability. This gives the seller an incentive to ``undercut across the gap'' to increase her chances of winning. Consider a realized profile of opponents' bids in which the lowest bid among the opponents equals \(\underline b_i\) and exactly \(k\ge 1\) opponents bid \(\underline b_i\). If she bids \(\underline b_i\), she wins with probability \(1/(k+1)\), and when she wins her payment is \(\underline b_i\), so her expected payoff in this contingency is \((\underline b_i-q)/(k+1)\). Submitting a bid \(b\le \overline b_{i-1}\) instead makes her the unique lowest bidder in the contingency above. Under a standard second-price payment, she would still be paid \(\underline b_i\), which would yield payoff \(\underline b_i-q > (\underline b_i-q)/(k+1)\). The payment reduction rule lowers the winner's payment exactly in these contingencies, eliminating this underbidding incentive and restoring the weak dominance of the strategy described in \autoref{l:wd}.

To pin down the reduced payment \(p\) in such cases, we need to ensure that
\[
p-q \le \frac{\underline b_i-q}{k+1} \quad \text{for all } q \in [\overline b_{i-1},\underline b_i),
\]
or equivalently,
\[
p \le \frac{\underline b_i+k q}{k+1} \quad \text{for all } q \in [\overline b_{i-1},\underline b_i).
\]
Because the right-hand side is strictly increasing in \(q\), it suffices that the inequality holds at \(q=\overline b_{i-1}\), i.e., \(p \le (\underline b_i+k\overline b_{i-1})/(k+1)\). Conversely, the reduced payment cannot be chosen strictly below this value: if \(p < (\underline b_i+k\overline b_{i-1})/(k+1)\), there exists \(\varepsilon>0\) such that a seller with quality \(q\in(\overline b_{i-1}-\varepsilon,\overline b_{i-1})\) strictly prefers to overbid to \(\underline b_i\) in the contingency above. Therefore, weak dominance pins down the reduced payment uniquely as \(p=(\underline b_i+k\overline b_{i-1})/(k+1)\).

We are now ready to present the main result of this section. A function \(h\) is said to be \textbf{structured} if there exists a finite partition of \([0,1]\) into intervals on which \(h\) is either increasing or decreasing. If the virtual surplus \(g\) is structured, then there exist \(L \in \mathbb{N}\) and a collection of disjoint intervals \(\{[\underline{s}_i, \overline{s}_i)\}_{i = 1}^{L}\) such that \(\overline{G}\) is affine on each interval \([\underline{s}_i, \overline{s}_i)\) for \(i=1, \ldots, L\) and coincides with \(G\) on the complement \([0,\overline{S}] \big\backslash \bigcup_{i =1}^{L}\left[\underline{s}_{i}, \overline{s}_{i}\right)\), where \(\overline{S}\), the optimal exclusion cutoff, is defined in \eqref{ivs}. We refer to the intervals \(\{[\underline{s}_i, \overline{s}_i)\}_{i = 1}^{L}\) as \emph{pooling intervals}, since the optimal interim allocation is constant on each of them, and to the intervals comprising \([0,\overline{S}] \big\backslash \bigcup_{i =1}^{L}\left[\underline{s}_{i}, \overline{s}_{i}\right)\) as \emph{non-pooling intervals}. It follows that there are at most \(L+1\) non-pooling intervals.

\begin{theorem}\label{t:bra_bo}
Suppose the buyer's virtual surplus \(g\) is structured. 
A buyer-optimal procurement mechanism is a BRA with bid intervals
\(\{[\underline b_i,\overline b_i]\}_{i=1}^M\) defined as follows:
\begin{itemize}
    \item If \(\underline{s}_1=0\), then \(M=L\) and
    \[
    [\underline b_i,\overline b_i]=
    \begin{cases}
    \big[F^{-1}(\overline s_i),\,F^{-1}(\underline s_{i+1})\big], & i=1,\ldots,L-1,\\[2pt]
    \big[F^{-1}(\overline s_L),\,F^{-1}(\overline S)\big], & i=L.
    \end{cases}
    \]
    \item If \(\underline{s}_1>0\), then \(M=L+1\) and
    \[
    [\underline b_i,\overline b_i]=
    \begin{cases}
    \big[0,\,F^{-1}(\underline s_1)\big], & i=1,\\[2pt]
    \big[F^{-1}(\overline s_{i-1}),\,F^{-1}(\underline s_i)\big], & i=2,\ldots,L,\\[2pt]
    \big[F^{-1}(\overline s_L),\,F^{-1}(\overline S)\big], & i=L+1.
    \end{cases}
    \]
\end{itemize}
\end{theorem}

Although a standard second-price procurement auction is generally not optimal in our setting due to quality concerns, \autoref{t:bra_bo} shows that restricting bids to suitably chosen intervals renders a simple second-price-like auction format optimal. The construction in \autoref{t:bra_bo} is chosen so that each non-pooling interval of the optimal interim allocation corresponds to a bid interval, while each pooling interval corresponds to a gap between two bid intervals. Consequently, a seller with a quality quantile in a non-pooling interval bids her quality as in a standard second-price auction, and a seller with a quality quantile in a pooling interval bids the highest quality of that interval. 

\autoref{l:wd} implies that if a seller's quality quantile \(s\) is in a non-pooling interval, she wins if and only if her quality quantile is the lowest; i.e., every other seller's quality quantile exceeds \(s\). Therefore, her quantile interim allocation equals \mbox{\(\hat P(s)=(1-s)^{n-1}\)} on non-pooling intervals. Furthermore, sellers with quality quantiles in the same pooling interval submit the same bid regardless of their exact quality quantiles, which indicates that the interim allocation probability is constant on that interval. Finally, if a seller has a quality quantile above \(\overline{S}\), she does not bid and hence never wins. In the proof, we formally show that the interim allocation induced by the BRA described in \autoref{t:bra_bo} coincides with the optimal interim allocation \eqref{eq:oip} in \autoref{p:oia}, which establishes its optimality.

In certain settings, the optimal BRA takes simple forms.

\begin{corollary} \label{cor:bra_sc}%~
    \begin{enumerate}[label={(\roman*})]
        \item \label{cor:mono} \citep[][]{mv95} 
        \begin{itemize}
            \item If the virtual surplus is decreasing, the BRA with one bid interval \([0,F^{-1}(\overline{S})]\) is optimal, which is equivalent to a standard second-price auction with reserve price \(F^{-1}(\overline{S})\).
            \item If the virtual surplus is increasing, the BRA with a degenerate bid interval \(\{1\}\) is optimal, which is equivalent to selecting the winning seller at random. 
        \end{itemize}

        \item \citep{lpv} If the virtual surplus is single-peaked, the BRA with one bid interval \([\underline{b}, F^{-1}(\overline{S})]\) is optimal. \label{cor:sp}

        \item If the virtual surplus is single-dipped, the BRA with two bid intervals \([0, \overline{b}]\) and \(\{F^{-1}(\overline{S})\}\) is optimal. \label{cor:sd}
    \end{enumerate}
\end{corollary}

\autoref{cor:bra_sc} \ref{cor:mono} recovers the symmetric-case results of \cite{mv95}: when quality concerns are globally mild---equivalently, the buyer's virtual surplus is decreasing---the buyer optimally relies on a standard second-price procurement auction with an appropriate reserve price. When quality concerns are sufficiently severe---equivalently, when the buyer's virtual surplus is increasing---it becomes optimal for the buyer to completely abandon competitive bidding and randomly select the winning seller. In our framework, these arise as two polar cases of the BRA: no restrictions on bidding except for the reserve price (so there is only one bid interval, which is the entire relevant interval \([0,F^{-1}(\overline{S})]\)), and the opposite extreme in which only one bid is allowed.

\autoref{cor:bra_sc} \ref{cor:sp} matches the main result in \cite{lpv}: when the virtual surplus is single-peaked, an optimal procurement mechanism is a lowball lottery auction (LoLA), which is a 
second-price auction with both a floor price and a reserve price. This corresponds exactly to a BRA with a single bid interval, allowing for a gap \emph{below} the bid interval: the gap ``flattens'' competition among the sellers offering lower-quality goods (where quality concerns are most acute), while sellers with goods at quality levels above the gap compete as in a standard procurement auction. Finally, \autoref{cor:bra_sc} \ref{cor:sd} shows that the optimal mechanism remains simple when the virtual surplus is single-dipped.

Note that the payment reduction rule does not arise in \cite{mv95} or \cite{lpv}, since single-peakedness implies a single bid interval. In more general settings, however, the payment reduction rule is necessary to deter underbidding by sellers whose quality falls within a gap between two bid intervals. As \autoref{example:reli} illustrates, it remains relevant even when the virtual surplus is single-dipped.

\begin{example} \label{example:reli}
Suppose there are two potential sellers whose costs are identically, independently, and uniformly distributed on \([0,1]\). Interpret \(q\) as a \emph{costly reliability index}: higher \(q\) is more expensive to supply but also yields a more reliable good (e.g., a machine, medical device, or infrastructure component).% whose durability is learned only over time

The buyer obtains a flow payoff \(\kappa>0\) per unit of time until the procured good breaks down.
Conditional on quality \(q\), the failure time \(T\) is exponentially distributed with rate \(\lambda(q) \coloneqq \hat\lambda-q\), where \(\hat\lambda>1\), and the buyer discounts at rate \(\delta\ge 0\). 
Therefore, the buyer's valuation is
\[
v(q)=\mathbb E\left[\int_0^T \kappa e^{-\delta t}dt \,\bigg| \, q\right]
=\frac{\kappa}{\delta+\lambda(q)} = \frac{\kappa}{\delta+\hat\lambda-q} = \frac{\kappa}{\Lambda-q},
\]
where \(\Lambda \coloneqq \delta+\hat\lambda>1\),
and her virtual surplus is 
\[
g(q)=v(q)-q-\frac{F(q)}{f(q)}
=\frac{\kappa}{\Lambda-q}-2q.
\]

To be more concrete, take \(\Lambda=1.33\) (e.g., \(\delta=0.03\) and \(\hat\lambda=1.3\)) and \(\kappa=1\); in this case, \(g\) is single-dipped.\footnote{More generally, \(g\) is single-dipped whenever \(2(\Lambda-1)^2<\kappa<2\Lambda^2\).} 
We plot \(g\) and the ironed virtual surplus, \(\overline{g}\), in \autoref{fig:relia}. By \autoref{p:oia}, the optimal interim allocation is given by
    \[
    \hat{P}(q) = 
    \left\{\begin{array}{ll}
    	1-q & q < 0.346 \\
    	0.327 & q \ge 0.346,
    \end{array}\right.
    \]
    which is shown in \autoref{fig:relib}.
By \autoref{cor:bra_sc} \ref{cor:sd}, a buyer-optimal mechanism is a BRA with two bid intervals, \([0,0.346]\) and \(\{1\}\). The purpose of excluding bids in \((0.346,1)\) is to weaken price competition among sellers whose qualities lie in this region. This restriction raises the buyer's expected payoff because, absent it, whenever both sellers' qualities fall within \([0.346,1]\) the second-price auction will always select the lower-quality seller, which \autoref{fig:relia} indicates would lead to a low virtual surplus. 

\begin{figure}[ht]%
    \begin{minipage}{.47\linewidth} 
        \centering
        \subfloat[]{
        \begin{tikzpicture}[xscale=5,yscale=4, >=Stealth]
        	\draw[thick, ->] (0,0) node[below left] {\small \(O\)} -- (1.15,0) node[below] {\small \(q\)};
        	\draw[thick, ->] (0,0)-- (0,1.15);

            % g(q) = v(q) - 2q with v(q) = 1/(1.33 - q)
        	\draw[DarkBlue, thick, domain=0:1] plot (\x, {1/(1.33 - \x) - 2*\x});

            % Ironing cutoff (computed numerically): b \approx 0.346
            % Ironed value: g(b) \approx 0.324
        	\draw[orange,thick] (0.346,0.324) -- (1,0.324);
        	\draw[orange,thick,domain=0:0.346] plot (\x, {1/(1.33 - \x) - 2*\x});

        	\node[below] at (1,0) {\small 1};
        	\node[above, orange] at (0.25,0.55){\small \(\overline{g}(q)\)};
        	\node[above left, DarkBlue] at (0.93,0.5){\small \(g(q)\)};

        	\draw[dotted, thick] (0.346,0.324) -- (0.346,0) node[below]{\small 0.346};
        \end{tikzpicture}
        \label{fig:relia}}
    \end{minipage}
    \begin{minipage}{.47\linewidth}
        \centering
        \subfloat[]{
        	\begin{tikzpicture}[xscale=5,yscale=4, >=Stealth]
        	\draw[thick, ->] (0,0) node[below left] {\small \(O\)} -- (1.15,0) node[below] {\small \(q\)};
        	\draw[thick, ->] (0,0)-- (0,1.15);

        	\draw[DarkBlue, thick, domain=0:1] plot (\x, {1-\x});
        	\node[above right, DarkBlue] at (0.5,0.5) {\small \(P^*(q)=1-q\)};

        	\node[left] at (0,0.327) {\small 0.327};

        	\draw[red,thick] (0.346,0.327) -- (1,0.327);
        	\draw[red,thick] (0,1) -- (0.346,0.654);
        	\node[above, red] at (0.3,0.75) {\small \(\hat{P}(q)\)};

        	\draw[dotted, thick] (0.346,0.654) -- (0.346,0) node[below]{\small 0.346};
        	\draw[dotted, thick] (1,0.327) -- (1,0) node[below] {\small 1};
        	\node[left] at (0,1) {\small 1};
    	\end{tikzpicture}      
        \label{fig:relib}}
    \end{minipage}
 \caption{\small In panel (a), the blue curve is the buyer's quantile virtual surplus \(\widetilde g\) (here \(\widetilde g(s)=g(q)\) since \(q\) is uniformly distributed), and the orange curve is the ironed quantile virtual surplus \(\overline g\). In panel (b), the blue curve is \(P^*(q)=1-q\) that appears in Border's condition, and the red curve is the optimal interim allocation \(\hat P(q)\).}%
 \label{fig:reli_example}%
\end{figure}

Recall the \emph{payment reduction rule} in \autoref{def:bra}: when the winning bid is the only bid in its bid interval \([\underline{b}_i,\bar b_i]\) and the second-lowest bid equals the lower bound \(\underline{b}_j\) of some higher bid interval \(j>i\), with \(k\) other sellers also bidding \(\underline{b}_j\). Then the winner is paid \((\underline{b}_j+k\,\bar b_{j-1})/(k+1)\), which is (weakly) below the second-lowest bid \(\underline{b}_j\).

In this example, the optimal BRA has two bid intervals: \([0,0.346]\) and \(\{1\}\). Consider the event in which one seller bids in \([0,0.346]\) (and hence wins) and the other seller bids \(1\). Here the winning bid is the only bid in the lower interval, and the second-lowest bid equals the lower bound of the higher interval. Since there is exactly one other seller bidding \(1\) (so \(k=1\)), the payment reduction rule sets the winning seller's payment to \((1+0.346)/2 = 0.673 < 1\) rather than paying the second-lowest bid \(1\), as in a standard second-price auction.

This reduction is exactly what deters sellers whose qualities lie in \([0.346,1]\) from underbidding just below \(0.346\). Suppose both sellers' qualities are in \( [0.346,1]\), so in equilibrium both bid \(1\) and each seller wins with probability \(1/2\), earning expected payoff \((1-q)/2\). If a seller with \(q\in[0.346,1]\) deviates to a bid just below \(0.346\), she wins for sure, but her payment is reduced to \(0.673\), yielding expected payoff \(0.673-q\). The deviation is unprofitable because \(0.673-q \le (1-q)/2\) whenever \(q\ge 0.346\).

Note that the payment reduction rule is \emph{inactive} in all other cases: if both bids lie in \([0,0.346] \), the winning seller is paid the second-lowest bid as usual; and if both bids equal \(1\), the winning seller is chosen with equal probability and paid \(1\).

Finally, the buyer's expected payoff under the buyer-optimal BRA is about 0.448; by comparison, the standard second-price auction and random assignment achieve payoffs of 0.413 and 0.394, respectively. In fact, the optimal LoLA mechanism in this example is a standard second-price auction. Thus, the buyer-optimal BRA improves the buyer’s expected payoff by at least 8.3 percent relative to the three benchmark mechanisms. \hfill \(\lozenge\)
\end{example}
 
\section{The general problem} \label{sop}
In this section, we consider the more general objective of maximizing any weighted average of the buyer's expected payoff and the social surplus. We show that under certain regularity conditions, a BRA remains optimal, although a random reserve price may need to be introduced when the constraint that the buyer's expected payoff be non-negative binds.

If individual rationality for the buyer, \(\pi_b \ge 0\), need not hold---that is, the buyer's expected payoff can be negative\footnote{This case may be relevant to, for example, government procurement settings, where the government agency prioritizes social benefit and conducts procurement even when a cost-benefit analysis suggests that its expected cost may exceed the expected benefit.}---then the problem of finding the quantile interim allocation that maximizes the weighted average, \eqref{problem:weighted}, can be simplified to 
\[\max_{\widetilde{P} \in \Omega_w(\widetilde{P}^*)}~~ \int_{0}^{1} h_\gamma(F^{-1}(s)) \widetilde{P}(s) \, \mathrm{d}s.\]
The analysis in \autoref{section:buyer-optimal} implies that an analogue of \autoref{t:bra_bo} then holds:

\begin{customthm}{1*} \label{thm:star}
    Suppose the buyer's expected payoff is allowed to be negative. If the weighted virtual surplus \(h_{\gamma}\) is structured, then an optimal procurement mechanism is a BRA whose bid intervals are determined as in \autoref{t:bra_bo}, with \(h_{\gamma}\) in place of \(g\). 
\end{customthm}

However, when the buyer's individual rationality constraint is imposed, the problem becomes less straightforward. \autoref{ex:buyer-ir} shows that imposing the buyer's individual rationality constraint may drastically change the optimal mechanism and a slightly modified version of the BRA may be required.

\begin{example}\label{ex:buyer-ir}
	Suppose there are two potential sellers whose qualities are identically, independently, and uniformly distributed on \([0,1]\). The buyer's valuation of the good is given by \(v(q)=2.6\,q-2.85\,q^2+2.25\,q^3\), and the designer wishes to maximize the expected social surplus. If the buyer's expected payoff can be negative, the designer will choose an interim allocation that solves	
	\[\max_{\widetilde{P} \in \Omega_w(\widetilde{P}^*)}~~ \int_{0}^{1} h_0(q) \widetilde{P}(q) \, \mathrm{d}q,\]
	where \(h_0(q) \coloneqq v(q) - q = 1.6\,q-2.85\,q^2+2.25\,q^3\). Since \(h_0\) is increasing, \autoref{cor:bra_sc} \ref{cor:mono} implies that a BRA with a degenerate bid interval \(\{1\}\) is optimal, which is equivalent to random allocation. However, it can be calculated that the buyer's expected payoff is \(\pi_b = -0.0875 < 0\).
	
	Now suppose we impose the individual rationality constraint on the buyer; that is, \(\pi_b \ge 0\). The designer's problem then becomes
	\begin{align*}
		\max_{\widetilde{P} \in \Omega_w(\widetilde{P}^*)}~~ & \int_{0}^{1} h_0(q) \widetilde{P}(q) \, \mathrm{d}q \\
		\text{s.t.}~~~~~~ & \int_{0}^{1} g(q) \widetilde{P}(q) \, \mathrm{d}q \ge 0,
	\end{align*}
    where \(g(q) \coloneqq v(q) - 2q = 0.6\,q-2.85\,q^2+2.25\,q^3\); \(g\) has a peak followed by a trough and is non-positive on \([4/15,1]\).
	The unconstrained optimal mechanism discussed above is infeasible here since the buyer's expected payoff is strictly negative. Using a Lagrangian approach, one can show that an optimal interim allocation is given by
	\[
P^*(q)=
\begin{cases}
1-\dfrac{a^*}{2}, & 0\le q\le a^*,\\[6pt]
1-q, & a^*<q\le b^*,\\[6pt]
\bar P^*, & b^*<q\le 1,
\end{cases}
\]
where \(a^*=0.287\), \(b^* = 0.4335\), and \(\bar P^*=0.0084\).

In this example, absent the non-negativity constraint, the social surplus--maximizing mechanism yields \(\pi_b<0\), so once \(\pi_b\ge 0\) is imposed, the constraint binds at the optimum. Because \(h_0\) is increasing, social surplus favors admitting higher qualities, but including the highest qualities can be too costly for the buyer's expected payoff. In a setting where a deterministic exclusion cutoff can be adjusted continuously, the cutoff will be set so that \(\pi_b=0\). However, our setting often requires pooling over an interval of qualities, and hence a deterministic cutoff cannot exclude only a part of such an interval; it can only either include it or exclude it entirely. In particular, for the interval \((b^*,1]\), fully admitting it violates \(\pi_b\ge 0\), while fully excluding it sacrifices too much social surplus. Thus, the optimal interim allocation is a constant \(\bar P^*\) on \((b^*,1]\) but strictly less than the interim allocation probability of \((1-b^*)/2 = 0.28325\) under standard pooling. This motivates a device that reduces, but does not eliminate, allocation in that region.

We implement \(P^*\) by introducing an extra bid \(1\) in addition to the regular bid interval \([a^*,b^*]\). Regular bids compete as in a standard BRA, while sellers with qualities above \(b^*\) submit the extra bid \(1\). After bids are submitted, with probability \(\zeta\) all bids equal to \(1\) are treated as valid, and otherwise all such bids are ignored. Payments follow the BRA logic, except when a seller wins with a regular bid in \([a^*,b^*]\) and the second-lowest valid bid is the extra bid \(1\): in a standard BRA where \(\{1\}\) is an ordinary highest bid interval, the payment reduction rule will set the payment to \((1+b^*)/2\), whereas here it is scaled by \(\zeta\), yielding \(\zeta(1+b^*)/2\). This scaling is needed because a seller with quality \(q>b^*\) who underbids into \([a^*,b^*]\) is admitted for sure, rather than being admitted only with probability \(\zeta\), so the payment rule must discount the effect of a runner-up bid of \(1\) to keep such underbidding unprofitable. Consequently, a seller with \(q>b^*\) can win only when bids of 1 are treated as valid and the opponent also bids \(1\), in which case the tie is broken with equal probability. This yields the interim allocation probability \(P(q)=\zeta(1-b^*)/2\) for all \(q>b^*\); choosing \(\zeta=2\bar P^*/(1-b^*)\) makes this probability equal to \(\bar P^*\).
%We implement \(P^*\) by allowing admissible bids in the regular interval \([a^*,b^*]\) together with a single extra bid \(1\). Sellers with \(q\le a^*\) submit \(a^*\), sellers with \(q\in(a^*,b^*]\) submit their qualities, and sellers with \(q>b^*\) submit the extra bid \(1\). 
\hfill \(\lozenge\)
\end{example}

More generally, the non-negativity constraint in the general problem \eqref{problem:weighted} plays a role only when the unconstrained solution would yield a strictly negative expected payoff for the buyer; otherwise, the characterization from \autoref{thm:star} applies directly. When the non-negativity constraint binds, which is especially likely when the information rents are large (captured by \(F/f\)) or the weight on the buyer's expected payoff, \(\gamma\), is small,\footnote{This is because, relative to the buyer's payoff, the weighted objective places only weight \(\gamma\) on information rents.} the optimal mechanism generally involves excluding highest qualities, at which trade remains valuable for the weighted objective but is too costly for the buyer's expected payoff. In some environments, the constraint can be met by tightening a deterministic reserve price (equivalently, lowering the highest admitted quality). The optimal mechanism remains a standard BRA, with bid intervals determined by the same construction as before, but with the exclusion cutoff (\(\overline{b}_M\)) shifted downward to satisfy the non-negativity constraint. However, a difficulty arises when the exclusion cutoff that binds \(\pi_b = 0\) falls within a pooling region, where all sellers whose qualities fall in this region must be treated identically. This may occur when the effective virtual surplus (which adjusts the weighted virtual surplus to reflect the binding constraint) is locally increasing. As illustrated in \autoref{ex:buyer-ir}, introducing an extra high bid on top of a standard BRA allows for admitting this pool of qualities only with a specific probability, which binds the non-negativity constraint without completely discarding the surplus generated by these trades.

We formally define the needed modification of the BRA (\autoref{def:abra}) in \autoref{pso}, and we call the resulting mechanism the \emph{augmented bid-restricted auction} (aBRA for short). In brief, we modify a BRA in three ways to obtain an aBRA. First, we may need to introduce an extra bid \(B\) that is strictly larger than the upper endpoint of the highest bid interval, \(\bar{b}_M\). Second, if the extra bid \(B\) is introduced, any seller who submits \(B\) is qualified to participate in the auction only with probability \(\zeta \in (0,1)\). Finally, if the second-lowest (valid) bid is \(B\), the winning seller receives a reduced payment that is further adjusted by the qualification rate \(\zeta\) relative to the reduced payment in a standard BRA, ensuring that no seller with quality \(q \in (\bar{b}_M, B]\) has an incentive to bid below \(\bar{b}_M\) to take advantage of the guaranteed qualification. 

An aBRA can also be interpreted as a BRA with \(M+1\) bid intervals: \(\{[\underline{b}_i, \overline{b}_i]\}_{i=1}^{M}\) and \(\{B\}\), with a slightly adjusted payment rule and a \emph{random reserve price}. The random reserve price effectively takes the value \(B\) with probability \(\zeta\) and takes the value \(\overline{b}_M\) with complementary probability. This reserve price is only realized after all sellers have submitted their bids.

In \autoref{pso}, we formally show that under mild regularity conditions, an aBRA (in which the extra bid/random reserve price is introduced if necessary) maximizes any weighted average of the buyer's expected payoff and the social surplus, subject to the constraint that the buyer's expected payoff remains non-negative (\autoref{t:bra_wo}). We also identify environments in which some simple trading mechanisms are optimal (\autoref{cor:spa} and \autoref{cor:single_peaked}).

\section{Conclusion} \label{s:conclusion}
In this study we explored procurement design problems in which the buyer's valuation of the good supplied depends directly on its quality, which is both unverifiable and unobservable. We analyzed the problem of maximizing an arbitrary weighted average of the buyer's expected payoff and the social surplus, subject to the constraint that the buyer's expected payoff remains non-negative. To tackle this problem, we employed a novel reduced-form approach utilizing techniques from linear optimization under a majorization constraint. We found that a \emph{bid-restricted auction}---a mechanism similar to a second-price auction, featuring a dominant strategy equilibrium but restricting sellers to bidding within specified intervals---is optimal.

In our analysis, we abstracted from collusion, repeated interaction, the possibility of defaults, and endogenous seller entry to isolate the role of quality concerns in procurement. Exploring the optimal procurement mechanism in the presence of these concerns could be an interesting direction for future research. It would also be valuable to extend the model to allow for ex-ante asymmetric sellers, to study comparative statics as the buyer's quality concerns intensify, and to incorporate additional design constraints (e.g., limits on monetary expenditures such as those imposed by anti-deficit rules). We leave these extensions for future work.

\begin{singlespace}
	\addcontentsline{toc}{section}{References}
	\bibliographystyle{aea}
	\bibliography{procurement.bib}
\end{singlespace} 
\addtocontents{toc}{\protect\setcounter{tocdepth}{1}} 

\begin{appendices}
\section{Results on majorization} \label{majorization_results}
The following results are taken from \cite{kms} and modified to our environment. Let \(A\) be an arbitrary subset of a topological vector space; we denote the set of its extreme points by \(\mathrm{ext} A\). Denote the set of decreasing functions in \(L^1\) that are majorized by \(f \in L^1\) by
\[
\Omega(f) := \{g \in L^1 : g \text{ is decreasing, } g \prec f\};
\]
similarly, denote %the ``weak majorization set'' by 
\(\Omega_w(f) := \{g \in L^1 : g \text{ is decreasing, } g \prec_w f\}\).

\begin{theorem}[Theorem 1 in \citeauthor{kms}, \citeyear{kms}] \label{kt1}
    Let \(f \in L^1\) be decreasing. Then \(h \in \mathrm{ext}\Omega(f)\) if and only if there exists a countable collection of disjoint intervals \(\left[\underline{x}_{i}, \overline{x}_{i}\right)\) indexed by \(i \in I\) such that for almost all \(x \in [0,1]\),
    \[
    h(x)=\left\{\begin{array}{ll}f(x) & \text { if } x \notin \bigcup_{i \in I}\left[\underline{x}_{i}, \overline{x}_{i}\right) \\ \frac{\int_{\underline{x}_{i}}^{\overline{x}_{i}} f(s) \mathrm{d} s}{\overline{x}_{i}-\underline{x}_{i}} & \text { if } x \in\left[\underline{x}_{i}, \overline{x}_{i}\right).\end{array}\right.
    \]
\end{theorem}

For \(B \subseteq [0,1]\), denote by \(\mathbf{1}_{B}(x)\) the indicator function of \(B\): it equals 1 if \(x \in B\) and 0 otherwise.

\begin{corollary}[Corollary 2 in \citeauthor{kms}, \citeyear{kms}] \label{wmc}
    Let \(f \in L^1\) be decreasing. Then \(h \in \mathrm{ext} \Omega_w(f)\) if and only if there exists \(\theta \in [0,1]\) such that \(h \in \mathrm{ext}\Omega(f \cdot \mathbf{1}_{[0,\theta]})\) and \(h(x)=0\) for almost all \(x \in (\theta, 1]\).
\end{corollary}

Now consider the problem
\begin{equation} \label{lpm}
\max_{m \in \Omega(f)} \int_{0}^{1} c(x) m(x) \, \mathrm{d}x,
\end{equation}
where \(f \in L^1\) is strictly decreasing, and \(c\) is a bounded function. Define 
\(C(x) \coloneqq \int_{0}^{x} c(s) \, \mathrm{d}s\),
and let \(\overline{C}\) be its concave hull. \autoref{pmm} characterizes a solution to problem \eqref{lpm}. 

\begin{proposition}[Proposition 2 in \citeauthor{kms}, \citeyear{kms}] \label{pmm}
    Let \(h \in \mathrm{ext}\Omega(f)\), and let \(\left\{\left[\underline{x}_{i}, \bar{x}_{i}\right) : i \in I\right\}\) be the collection of intervals described in \autoref{kt1}. Then \(h\) is optimal if and only if \(\overline{C}\) is affine on \(\left[\underline{x}_{i}, \bar{x}_{i}\right)\) for each \(i \in I\) and \(\overline{C} = C\) otherwise.
\end{proposition}

\section{Proofs for \autoref{section:buyer-optimal}} \label{buyer_optimal_proofs}
\subsection{Proof of \autoref{p:oia}} \label{pia}
Because the objective function of problem \eqref{problem:buyer} is linear, by Bauer's maximum principle (\cite{ab06}, Theorem 7.69, page 298), the maximum is attained at an extreme point \(\hat{P}\) of \(\Omega_w(\widetilde{P}^*)\). By \autoref{wmc}, there exists \(\bar{s} \in [0,1]\) such that \(\hat{P}\) is an extreme point of \(\Omega(\widetilde{P}^{*} \cdot \mathbf{1}_{[0, \bar{s}]})\) and vanishes on \([\bar{s}, 1]\). Furthermore, the optimality of \(\hat{P}\) requires \(\overline{g}(\bar{s}) = 0\); thus, setting \(\bar{s} = \sup\{s \in [0,1] : \overline{g}(s) \ge 0\} = \overline{S}\) suffices. 
Then \autoref{pmm} implies that
\[
    \hat{P}(s) = 
    \begin{cases}
    (1-s)^{n-1} & \text{ if } s \in [0,\overline{S}] \big\backslash \bigcup_{i \in \mathcal{I}} \left[\underline{s}_{i}, \overline{s}_{i}\right), \\
    \frac{\int_{\underline{s}_{i}}^{\overline{s}_{i}} (1-s)^{n-1} \mathrm{d}s}{\overline{s}_{i} - \underline{s}_{i}}& \text{ if } s \in \left[\underline{s}_{i}, \overline{s}_{i}\right), \\
    0 & \text{ if } s \in (\overline{S},1];
    \end{cases}
\]
where the collection \(\{[\underline{s}_i, \overline{s}_i) \subseteq [0,\overline{S}] : i \in \mathcal{I}\}\) is such that \(\overline{G}\) is affine on \([\underline{s}_i, \overline{s}_i)\) for each \(i \in \mathcal{I}\) and \(\overline{G} = G\) otherwise.

\subsection{Proof of \autoref{l:wd}}
First, consider a seller with \(q \ge \overline{b}_M\). Because the highest allowable bid is \(\overline{b}_M\), if she ever wins, her payoff is at most zero, which is no better than not bidding regardless of what other sellers do.

Next, consider a seller who has \(q \in (\underline{b}_i, \overline{b}_i)\) for some \(i=1, \ldots, M\); let \(b_{-}^{\min}\) denote the minimum bid among all other sellers. By bidding exactly \(q\), there are the following cases:
\begin{itemize}
    \item if \(b_{-}^{\min }<q\), the seller's payoff is 0;
    \item if \(b_{-}^{\min } \in\left(q, \bar{b}_i\right]\) or \(b_{-}^{\min } \in\left(\underline{b}_j, \bar{b}_j\right]\) for some \(j>i\), the seller's payoff is \(b_{-}^{\min }-q\);
    \item if \(b_{-}^{\text {min }}=\underline{b}_j\) for some \(j>i\), and there are \(k\) other sellers who bid \(b=b_{-}^{\min }\), then the seller's payoff is \((\underline{b}_j+k \overline{b}_{j-1})/(k+1)-q\).
\end{itemize}
If the seller bids \(b < q\) instead, when \(b \leq b_{-}^{\min }<q\) she gets a strictly negative payoff instead of 0, and otherwise she gets the same payoff as bidding \(q\). If the seller bids \(b > q\), when \(b_{-}^{\min } < b\) she gets zero, and when \(b_{-}^{\min } \ge b\) her payoff is otherwise identical to bidding \(b=q\), except for the case that \(b_{-}^{\min }=b=\underline{b}_j\) for some \(j>i\). In this case, if \(k\) other sellers bid \(\underline{b}_j\), then the seller's (expected) payoff is \((\underline{b}_j-q)/(k+1)\). But by bidding \(b=q\), fixing other sellers' bids, her payoff is \((\underline{b}_j+k \overline{b}_{j-1})/(k+1)-q\), and
\[\frac{\underline{b}_j+k \bar{b}_{j-1}}{k+1}-q-\frac{\underline{b}_j-q}{k+1}=\frac{k\left(\bar{b}_{j-1}-q\right)}{k+1} \ge 0,\]
where the inequality holds because \(i \le j-1\). Thus, for a seller with \(q \in (\underline{b}_i, \overline{b}_i)\) for some \(i=1, \ldots, M\), bidding \(q\) is a weakly dominant strategy. 

Now consider a seller with \(q \in [\overline{b}_{i-1}, \underline{b}_i]\) for some \(i=1, \ldots, M\) (recall that we set \(\overline{b}_0 := 0\)). By bidding \(\underline{b}_i\), there are the following cases:
\begin{itemize}
    \item if \(b_{-}^{\text {min }} \le \overline{b}_{i-1}\), the seller's payoff is 0;
    \item if \(b_{-}^{\text {min }}=\underline{b}_i\), the seller's payoff is \((\underline{b}_i-q)/(k+1)\) if \(k\) other sellers bid \(\underline{b}_i\);
    \item if \(b_{-}^{\min } \in \left(\underline{b}_j, \bar{b}_j\right]\) for some \(j \ge i\), the seller's payoff is \(b_{-}^{\min }-q\);
    \item if \(b_{-}^{\min }=\underline{b}_j\) for some \(j>i\), and there are \(k\) other sellers who bid \(b=b_{-}^{\min }\), then the seller's payoff is \((\underline{b}_j+k \overline{b}_{j-1})/(k+1)-q\).
\end{itemize}
If the seller bids \(b \le \overline{b}_{i-1}\) instead, she loses when \(b^{\min}_{-} < b\) and gets zero payoff; when \(b \le b_{-}^{\min } \le \bar{b}_{i-1}\), the seller's payoff is bounded above by zero. When \(b_{-}^{\min}=\underline{b}_i\), and \(k\) other sellers bid \(\underline{b}_i\), the seller gets \((\underline{b}_i+k \bar{b}_{i-1})/(k+1)-q\); but the payoff from bidding \(\underline{b}_i\) is \((\underline{b}_i-q)/(k+1)\), and 
\[\left(\frac{\underline{b}_i+k \bar{b}_{i-1}}{k+1}-q\right)-\frac{\underline{b}_i-q}{k+1}=\frac{k\left(\bar{b}_{i-1}-q\right)}{k+1} \le 0,\]
where the inequality holds because \(q \ge \bar{b}_{i-1}\). Otherwise, the seller gets the same payoff as bidding \(\underline{b}_i\). If the seller bids \(b > \underline{b}_{i}\), her payoff is 0 if \(b_{-}^{\min }<b\), and if \(b_{-}^{\min } \ge b\), the seller gets the same payoff as bidding \(\underline{b}_i\) unless \(b_{-}^{\min }=b=\underline{b}_j\) for some \(j>i\). In this case, if \(k\) other sellers bid \(\underline{b}_j\), the seller's payoff is \((\underline{b}_j - q)/(k+1)\); but in this case, by bidding \(\underline{b}_i\) the seller gets \((\underline{b}_j+k \overline{b}_{j-1})/(k+1)-q\), which is no lower. Therefore, for a seller with \(q \in [\overline{b}_{i-1}, \underline{b}_i]\) for some \(i=1, \ldots, M\), bidding \(\underline{b}_i\) is a weakly dominant strategy. 

Finally, notice that the argument above also establishes that bidding the quality is a weakly dominant strategy if the seller's quality is \(\underline{b}_i\) for some \(i=1, \ldots, M\). This completes the proof.

\subsection{Proof of \autoref{t:bra_bo}}
We first consider the case that \(\overline{S} = 0\), in which not buying from any seller is optimal. In this case, the optimal mechanism is a BRA with one (trivial) bid interval \(\{0\}\). We assume that \(\overline{S} > 0\) henceforth. 

The description of the BRA in the statement of the theorem and \autoref{l:wd} together indicate that 
\begin{itemize}
    \item If a seller's quality quantile \(s \in [0,\overline{S}) \big\backslash \bigcup_{i =1}^{L}\left[\underline{s}_{i}, \overline{s}_{i}\right)\) (that is, in a non-pooling interval), she bids her quality: \(b = q = F^{-1}(s)\).
    \item If a seller's quality quantile \(s \in \left[\underline{s}_{i}, \overline{s}_{i}\right)\) for some \(i = 1, \ldots, L\) (that is, in a pooling interval), she bids the lower bound of a bid interval.
    \item If a seller's quality quantile \(s \ge \overline{S}\), she does not bid. 
\end{itemize}
Therefore, for \(s \ge \overline{S}\), the interim allocation probability \(P_{BRA}(s)\) is zero. Moreover, by the definition of a BRA, a seller with quality quantile \(s \in [0,\overline{S}) \big\backslash \bigcup_{i =1}^{L}\left[\underline{s}_{i}, \overline{s}_{i}\right)\) wins if and only if all other sellers have quality quantiles above \(s\), which happens with probability \((1-s)^{n-1}\). Consequently, \(P_{BRA}(s) = (1-s)^{n-1}\) for all \(s \in [0,\overline{S}) \big\backslash \bigcup_{i =1}^{L}\left[\underline{s}_{i}, \overline{s}_{i}\right)\).

Now consider a seller with quality quantile \(s \in \left[\underline{s}_{i}, \overline{s}_{i}\right)\) for some \(i = 1, \ldots, L\). She could only win the auction when there is no seller with \(s < \underline{s}_{i}\); there are the following cases: 
\begin{itemize}
    \item All other \(n-1\) sellers have \(s \ge \overline{s}_{i}\). This case happens with probability \(\left(1-\overline{s}_i\right)^{n-1}\), and in this case, this seller wins with probability 1.

    \item One other seller has \(s \in \left[\underline{s}_{i}, \overline{s}_{i}\right)\), and \(n-2\) other sellers have \(s \ge \overline{s}_{i}\). This case happens with probability \(\binom{n-1}{1}\left(1-\overline{s}_i\right)^{n-2}\left(\overline{s}_i-\underline{s}_i\right)\), and in this case, this seller wins with probability 1/2.

    \item Two other sellers have \(s \in \left[\underline{s}_{i}, \overline{s}_{i}\right)\), and \(n-3\) other sellers have \(s \ge \overline{s}_{i}\). This case happens with probability \(\binom{n-1}{2}\left(1-\overline{s}_i\right)^{n-3}\left(\overline{s}_i-\underline{s}_i\right)^2\), and in this case, this seller wins with probability 1/3.

    \item ...

    \item \(n-2\) other sellers have \(s \in \left[\underline{s}_{i}, \overline{s}_{i}\right)\), and one other seller has \(s \ge \overline{s}_{i}\). This case happens with probability \(\binom{n-1}{n-2}\left(1-\overline{s}_i\right)\left(\overline{s}_i-\underline{s}_i\right)^{n-2}\), and in this case, this seller wins with probability \(1/(n-1)\).

    \item All other sellers have \(s \in \left[\underline{s}_{i}, \overline{s}_{i}\right)\). This case happens with probability \(\binom{n-1}{n-1}\left(\overline{s}_i-\underline{s}_i\right)^{n-1}\), and in this case, this seller wins with probability \(1/n\).
\end{itemize}
Therefore, the interim allocation probability for this seller with quality quantile \(s \in \left[\underline{s}_{i}, \overline{s}_{i}\right)\) is given by
\begin{align}
%& \left(1-\overline{s}_i\right)^{n-1} + \binom{n-1}{1}\frac{\left(1-\overline{s}_i\right)^{n-2}\left(\overline{s}_i-\underline{s}_i\right)}{2} + \binom{n-1}{2}\frac{\left(1-\overline{s}_i\right)^{n-3}\left(\overline{s}_i-\underline{s}_i\right)^2}{3} + \nonumber \\
%& \cdots + \binom{n-1}{n-2}\frac{\left(1-\overline{s}_i\right)\left(\overline{s}_i-\underline{s}_i\right)^{n-2}}{n-1} + \binom{n-1}{n-1}\frac{\left(\overline{s}_i-\underline{s}_i\right)^{n-1}}{n} = 
\sum_{k=0}^{n-1} \binom{n-1}{k} \frac{(1-\overline{s}_i)^{n-(k+1)}(\overline{s}_i-\underline{s}_i)^k}{k+1}. \label{eq:prob_total}
\end{align}
%Observing that for every \(k=0, \ldots, n-1\), 
%\[\frac{1}{k+1}\binom{n-1}{k} = \frac{1}{n}\binom{n}{k+1};\]
Using the identity \(\frac{1}{k+1}\binom{n-1}{k} = \frac{1}{n}\binom{n}{k+1}\), %the summation in (10) can be rewritten as:
\eqref{eq:prob_total} can be rewritten as
\begin{align*}
    \, & \frac{1}{n} \sum_{k=0}^{n-1} \binom{n}{k+1} (1-\overline{s}_i)^{n-(k+1)}(\overline{s}_i-\underline{s}_i)^k = \frac{1}{n(\overline{s}_i-\underline{s}_i)} \sum_{j=1}^{n} \binom{n}{j} (1-\overline{s}_i)^{n-j}(\overline{s}_i-\underline{s}_i)^j \\
    = \, & \frac{\left[(1-\overline{s}_i) + (\overline{s}_i-\underline{s}_i)\right]^n - (1-\overline{s}_i)^n}{n(\overline{s}_i-\underline{s}_i)} = \frac{(1-\underline{s}_i)^n - (1-\overline{s}_i)^n}{n(\overline{s}_i-\underline{s}_i)} = \frac{\int_{\underline{s}_{i}}^{\overline{s}_{i}} (1-s)^{n-1} \, \mathrm{d}s}{\overline{s}_{i} - \underline{s}_{i}},
\end{align*}
where the first equality is obtained by letting \(j=k+1\), the second equality follows from the binomial theorem, and the final equality follows from the fundamental theorem of calculus. Thus, \(P_{BRA}(s) = \frac{\int_{\underline{s}_{i}}^{\overline{s}_{i}} (1-s)^{n-1} \, \mathrm{d}s}{\overline{s}_{i} - \underline{s}_{i}}\) on \([\underline{s}_i, \overline{s}_i)\) for every \(i=1, \ldots, L\). 

Summing up, we have established that the interim allocation probability induced by the BRA described in the statement of \autoref{t:bra_bo}, \(P_{BRA}\), equals the interim allocation probability \(\hat{P}\) defined by \eqref{eq:oip} in \autoref{p:oia} \(F\)-almost everywhere, which implies that the BRA described in the statement of the theorem is indeed a buyer-optimal procurement mechanism. 

\section{Formal Analysis for \autoref{sop}} \label{pso}
\subsection{Optimal interim allocation} \label{soi}
To solve problem \eqref{problem:weighted}, we use a Lagrangian approach. Set up the Lagrangian with multiplier \(\lambda\): 
\[
\mathcal{L}_{\gamma}(\widetilde{P};\lambda) \coloneqq \int_{0}^{1} \left[\widetilde{h}_{\gamma}(s) + \lambda \widetilde{g}(s)\right] \widetilde{P}(s) \, \mathrm{d}s,
\]
where \(\widetilde{h}_{\gamma}(s) := h_{\gamma}(F^{-1}(s))\) is the quantile version of the weighted virtual surplus. Define \(\phi_{\gamma}(q; \lambda) := h_{\gamma}(q) + \lambda g(q)\), and let \(\widetilde{\phi}_{\gamma}(s;\lambda) := \phi_{\gamma}(F^{-1}(s); \lambda)\). Evidently, \(\widetilde{\phi}_{\gamma}(s;\lambda) = \widetilde{h}_{\gamma}(s) + \lambda \widetilde{g}(s)\), which is the quantile virtual surplus of the Lagrangian. We solve the problem by maximizing \(\mathcal{L}_\gamma(\,\cdot\,; \lambda)\) over \(\widetilde{P} \in \Omega_w(\widetilde{P}^*)\), and then find an appropriate Lagrangian multiplier \(\lambda^*\) so that the complementary slackness condition holds. 

We iron \(\widetilde{\phi}_{\gamma}(s;\lambda)\) to make sure that the monotonicity constraint holds: let 
\[
H_{\gamma}(s) \coloneqq \int_{0}^{s} \widetilde{h}_{\gamma}(x) \, \mathrm{d}x \quad \text{and} \quad \Phi_{\gamma}(s;\lambda) \coloneqq \int_{0}^{s} \widetilde{\phi}_{\gamma}(x; \lambda) \, \mathrm{d}x,
\]
and let \(\overline{H}_{\gamma}\) and \(\overline{\Phi}_{\gamma}\) be the concave hulls of \(H_{\gamma}\) and \(\Phi_{\gamma}\), respectively. Now let 
\[
\overline{h}_{\gamma}(s) \coloneqq \overline{H}_{\gamma}'(s), \quad \text{and} \quad \overline{\phi}_{\gamma}(s;\lambda) \coloneqq \frac{\partial}{\partial s} \,\overline{\Phi}_{\gamma}(s; \lambda);
\]
\(\overline{h}_{\gamma}(s)\) and \(\overline{\phi}_{\gamma}(s; \lambda)\) are the ironed quantile weighted virtual surplus and the ironed quantile virtual surplus of the Lagrangian, respectively. By construction, both are decreasing in \(s\). 

Similar to the buyer-optimal case in \autoref{p:oia}, for all \(s\) such that \(\overline{\phi}_{\gamma}(s;\lambda) > 0\), the optimal interim allocation \(\hat{P}\) is flat whenever ironing is needed, and coincides with \(\widetilde{P}^*(s) = (1-s)^{n-1}\) otherwise. If there exists an interval \([S_+,S_0]\) on which \(\overline{\phi}_{\gamma}(s;\lambda) = 0\), we may need to find some \(\overline{P}\) satisfying
\[
0 \le \overline{P} \le \frac{\int_{S_+}^{S_0} (1-s)^{n-1} \, \mathrm{d}s}{S_0-S_+}:=A(S_+, S_0)
\]
and set \(\hat{P}(s) = \overline{P}\) on \([S_+,S_0]\) to satisfy complementary slackness; the second inequality in the above expression is needed because Border's condition requires \(\hat{P} \prec_w \widetilde{P}^*\).

\begin{proposition}\label{p:wia}
    If there exist \(\lambda^* \ge 0\), \(0 \le S_+ \le S_0 \le 1\), and a collection of disjoint intervals \(\{\left[\underline{s}_{i}, \bar{s}_{i}\right)\}_{i \in \mathcal{I}}\) with \(\left[\underline{s}_{i}, \bar{s}_{i}\right) \subseteq [0,S_0]\) for each \(i \in \mathcal{I}\), such that 
    \begin{itemize}
        \item[(i)] \(S_+ = \sup\left\{s \in [0,1]: \overline{\phi}_{\gamma}(s;\lambda^*) > 0\right\}\) and \(S_0 = \sup\left\{s \in [0,1]: \overline{\phi}_{\gamma}(s;\lambda^*) \ge 0\right\}\),
        \item[(ii)] \(\overline{\Phi}_{\gamma}(s;\lambda^*)\) is affine on \(\left[\underline{s}_{i}, \bar{s}_{i}\right)\) for each \(i \in \mathcal{I}\) and \([S_+,S_0)\), and
        \item[(iii)] \(\overline{\Phi}_{\gamma}(s;\lambda^*) = \Phi_{\gamma}(s;\lambda^*)\) on \([0,S_+] \big\backslash  \bigcup_{i \in \mathcal{I}}\left[\underline{s}_{i}, \overline{s}_{i}\right)\),
    \end{itemize}
    then the interim allocation
    \begin{equation} \label{eq:wip}
    \hat{P}(s) = \left\{
    \begin{array}{ll}
    (1-s)^{n-1} & \text{if } s \in [0, S_+] \big\backslash \bigcup_{i \in \mathcal{I}} \left[\underline{s}_{i}, \overline{s}_{i}\right) \\
    \frac{\int_{\underline{s}_{i}}^{\overline{s}_{i}} (1-s)^{n-1} \mathrm{d}s}{\overline{s}_{i} - \underline{s}_{i}} & \text{if } s \in \left[\underline{s}_{i}, \overline{s}_{i}\right) \\
    \overline{P} & \text{if } s \in (S_+,S_0] \\
    0 & \text{if } s \in (S_0,1]
    \end{array}\right.
    %\label{iaf}
    \end{equation}
    with \(\overline{P} \in \left(0, A(S_+, S_0)\right]\) such that \(\hat{P}\) satisfies the complementary slackness condition
    \[
    \lambda^* \int_{0}^{1} \widetilde{g}(s) \hat{P}(s) \, \mathrm{d}s = 0
    \]
    is optimal.
\end{proposition}

\begin{proof}
Suppose first that \(\overline{h}_{\gamma}(0) < 0\), meaning that it is undesirable to trade under incomplete information. In this case, the BRA with one (trivial) bid interval \(\{0\}\) is optimal. Now suppose instead that \(\overline{h}_{\gamma}(0) \ge 0\). Define
\(Z := \left\{s \in [0,1] : \overline{h}_{\gamma}(s) = 0 \right\}\); \(Z\) is the set of points on which the ironed quantile social surplus is zero. Because \(\overline{h}_{\gamma}(0) \ge 0\), if \(Z\) is empty, we must have \(\overline{h}_{\gamma}(1) > 0\). 

By Corollary 1 on page 219 and Theorem 2 on page 221 of \cite{luenberger1969}, \(\hat{P} \in \Omega_w(\widetilde{P}^*)\) solves problem \eqref{problem:weighted} if and only if there exists \(\lambda \ge 0\) such that \(\hat{P}\) maximizes \(\mathcal{L}\), and the complementary slackness condition
\[
\lambda \int_{0}^{1} \widetilde{g}(s) \hat{P}(s) \, \mathrm{d}s = 0
\]
holds. Consequently, an optimal \(\hat{P}\) can be found using the following algorithm: 

\paragraph{Step 1.} Check whether there exists \(\hat{s} \in [0,1]\), either \(\hat{s} \in Z\), or \(Z = \varnothing\) and \(\hat{s} = 1\) such that 
\[
\hat{P}(s) = \left\{
\begin{array}{ll}
    (1-s)^{n-1} & \text{if } s \in [0, \hat{s}] \big\backslash \bigcup_{i \in I} \left[\underline{x}_{i}, \overline{x}_{i}\right) \\
    \frac{\int_{\underline{x}_{i}}^{\overline{x}_{i}} (1-s)^{n-1} \mathrm{d}s}{\overline{x}_{i} - \underline{x}_{i}} & \text{if } s \in \left[\underline{x}_{i}, \overline{x}_{i}\right) \\
    0 & \text{if } s \in (\hat{s},1]
\end{array}\right.
\]
satisfies
\[
\int_{0}^{1} \widetilde{g}(s) \hat{P}(s) \, \mathrm{d}s \ge 0,
\]
where \(\left\{\left[\underline{x}_{i}, \overline{x}_{i}\right)\right\}_{i \in I}\) is the collection of intervals on which \(\overline{H}_{\gamma}\) is affine.
If so, we can set \(\lambda = 0\), which implies that the ironed quantile virtual surplus of the Lagrangian coincides with the ironed quantile weighted virtual surplus: \(\overline{\phi}_{\gamma}(s;0) = \overline{h}_{\gamma}(s)\). Setting \(S_+ = S_0 = \hat{s}\), \(\hat{P}\) solves problem \eqref{problem:weighted}. If not, go to \textbf{Step 2}. 

\paragraph{Step 2.} We must have \(\lambda > 0\), otherwise we could have found an \(\hat{s}\) in \textbf{Step 1}. Now we search for \(\lambda > 0\) such that there exists a unique \(\tilde{s}\) such that \(\overline{\phi}_{\gamma}(\tilde{s}; \lambda) = 0\), and the ``induced interim allocation'' 
\[
\hat{P}(s) = \left\{
\begin{array}{ll}
    (1-s)^{n-1} & \text{if } s \in [0, \tilde{s}] \big\backslash \bigcup_{i \in J} [\underline{y}_{i}, \overline{y}_{i}) \\
    \frac{\int_{\underline{y}_{i}}^{\overline{y}_{i}} (1-s)^{n-1} \mathrm{d}s}{\overline{y}_{i} - \underline{y}_{i}} & \text{if } s \in [\underline{y}_{i}, \overline{y}_{i}) \\
    0 & \text{if } s \in (\tilde{s},1]
\end{array}\right.
\]
satisfies
\[
\int_{0}^{1} \widetilde{g}(s) \hat{P}(s) \, \mathrm{d}s = 0,
\]
where \(\{[\underline{y}_{i}, \overline{y}_{i})\}_{i \in J}\) is the collection of intervals on which \(\overline{\Phi}_{\gamma}\) is affine. If we can find such a \((\lambda, \tilde{s})\) pair, \(\hat{P}\) solves problem \eqref{problem:weighted}; if not, go to \textbf{Step 3}. 
\paragraph{Step 3.} There must exist an interval \([S_+,S_0] \subseteq [0,1]\) with \(S_+ < S_0\) such that \(\overline{\phi}_\gamma(\,\cdot\,; \lambda) = 0\) on \([S_+,S_0]\), and there exists \(\overline{P}\) with
\[
0 \le \overline{P} \le \frac{\int_{S_+}^{S_0} (1-s)^{n-1} \, \mathrm{d}s}{S_0 - S_+}
\]
such that
\begin{equation*}
\hat{P}(s) = \left\{
\begin{array}{ll}
    (1-s)^{n-1} & \text{if } s \in [0, S_+] \big\backslash \bigcup_{i \in J} [\underline{y}_{i}, \overline{y}_{i}) \\
    \frac{\int_{\underline{y}_{i}}^{\overline{y}_{i}} (1-s)^{n-1} \mathrm{d}s}{\overline{y}_{i} - \underline{y}_{i}} & \text{if } s \in [\underline{y}_{i}, \overline{y}_{i}) \\
    \overline{P} & \text{if } s \in (S_+,S_0] \\
    0 & \text{if } s \in (S_0,1]
\end{array}\right.
%\label{iaf}
\end{equation*}
satisfying
\[
\int_{0}^{1} \widetilde{g}(s) \hat{P}(s) \, \mathrm{d}s = 0
\]
solves problem \eqref{problem:weighted}. 
\end{proof}

\subsection{Implementation: Augmenting BRA if necessary} \label{sim}
If the optimal interim allocation \(\hat{P}\) defined in \eqref{eq:wip} satisfies \(S_+ = S_0\), it takes exactly the same form as \eqref{eq:oip} in \autoref{p:oia}. Therefore, its implementation is described by \autoref{t:bra_bo}, by replacing \(\overline{S}\) by \(S_0\). If, instead, \(S_+<S_0\), then to implement the optimal interim allocation we need to slightly modify the BRA. %In this modified version, if a seller with quality \(q \in (S_+,S_0]\) wins the auction, she supplies the good with probability \(\overline{P}/A(S_+, S_0)\).
The details of this modification are described below.

\begin{definition} \label{def:abra}
    An \textbf{augmented bid-restricted auction (aBRA)} is a sealed-bid auction with 
    \begin{itemize}
        \item \(M \in \mathbb{N}\) \textbf{standard bid intervals} \(\{[\underline{b}_i, \overline{b}_i]\}_{i=1}^{M}\), where \(\underline{b}_1 \ge 0\), \(\overline{b}_M < 1\), and for all \(i=1,\ldots, M\), \(\underline{b}_i \le \overline{b}_i\) and \(\overline{b}_i < \underline{b}_{i+1}\),
        \item an \textbf{extra bid} \(B\) with \(\overline{b}_M < B \le 1\), and 
        \item a \textbf{qualification rate} \(\zeta \in (0,1]\),
    \end{itemize}
    with the following rules:
    \begin{itemize}
        \item Any seller who wishes to participate must submit a bid \(b \in \cup_{i=1}^M[\underline{b}_i, \overline{b}_i] \cup \{B\}\).
        \item With probability \(1-\zeta\), all sellers who bid \(B\) are disqualified, meaning that their bids do not count.
        
    \item The seller whose (valid) bid is the lowest wins the auction; in the event of a tie, the winning seller is chosen with equal probability among the sellers submitting the lowest (valid) bid.
    
    \item If the winning bid is the only bid in its standard bid interval \([\underline{b}_i, \overline{b}_i]\) for some \(i = 1, \ldots, M\), and
    \begin{itemize}
        \item if, furthermore, the second-lowest (valid) bid equals \(\underline{b}_j\) for some \(j > i\) and a total of \(k\) sellers bid \(\underline{b}_j\), then the winning seller receives a payment of
        \((\underline{b}_j + k\overline{b}_{j-1})/(k+1)\);
        
        \item if, furthermore, the second-lowest (valid) bid equals \(B\) and a total of \(k\) sellers bid \(B\),\footnote{If there is no second-lowest valid bid (e.g., only one seller submits a valid bid), we define the second-lowest bid to be \(B\) if the winning bid is \(B\), and \(\overline{b}_M\) if the winning bid is in a standard bid interval. \label{fn:reserve}} 
        then the winning seller receives a payment of 
        \(\zeta(B + k\overline{b}_{M})/(k+1)\).
    \end{itemize}
    
    \item Otherwise, the winning seller receives a payment equal to the second-lowest (valid) bid.
    \end{itemize}
\end{definition}

We first establish the aBRA analogue of \autoref{l:wd}. 

\begin{lemma} \label{l:wd_abra}
    In an aBRA, it is a weakly dominant strategy for any seller to not bid if her quality exceeds \(B\), bid her quality when \(q \in [\underline{b}_i, \overline{b}_i)\) for some \(i=1, \ldots, M\), bid \(\underline{b}_{i}\) if \(q \in [\overline{b}_{i-1}, \underline{b}_{i})\) for some \(i=1, \ldots, M\), and bid \(B\) if \(q \in [\overline{b}_M, B)\).
\end{lemma}

\begin{proof}
    Because the highest allowable bid is \(B\), if a seller with \(q \ge B\) ever wins, her payoff is at most zero, which is no better than not bidding regardless of what other sellers do. 

    Now consider a seller with \(q \in [\overline{b}_M, B)\). By bidding \(B\), there are two cases (recalling that \(b_{-}^{\min}\) denotes the minimum bid among all other sellers): 
    \begin{itemize}
        \item if \(b^{\min}_{-} \le \overline{b}_{M}\), the seller's expected payoff is 0;
        \item if \(b^{\min}_{-} = B\), and \(k\) other sellers bid \(B\), the seller's expected payoff is \(\zeta(B-q)/(k+1)\). 
    \end{itemize}
    If the seller bids \(b \le \overline{b}_M\) instead, she is qualified with probability 1. If \(b^{\min}_{-} \le \overline{b}_{M}\), the seller's expected payoff is at most zero. Now suppose \(b^{\min}_{-} = B\), then there are two possibilities. With probability \(\zeta\), the opponents who bid \(B\) are qualified; the seller wins, and her payoff is \(\zeta(B+k\overline{b}_M)/(k+1) - q\). With probability \(1-\zeta\), the seller submits the only valid bid. By \Cref{fn:reserve}, the second-lowest bid is taken as \(\overline{b}_M\). Thus, her payoff is \(\overline{b}_M - q\). Consequently, the seller's expected payoff from bidding \(b \le \overline{b}_M\) is given by
    \[\zeta \left( \frac{\zeta(B+k\overline{b}_M)}{k+1} - q \right) + (1-\zeta)(\overline{b}_M - q). \]
    The gain from bidding \(B\) relative to bidding \(b \le \overline{b}_M\) is:
    \begin{align*}
    & \frac{\zeta(B-q)}{k+1} - \left[ \frac{\zeta^2(B+k\overline{b}_M)}{k+1} - \zeta q + (1-\zeta)(\overline{b}_M - q) \right] \\
    &= \frac{\zeta B (1-\zeta) - \zeta^2 k \overline{b}_M + \zeta k q}{k+1} + (1-\zeta)(q - \overline{b}_M).
\end{align*}
Since \(q \ge \overline{b}_M\), the term \((1-\zeta)(q-\overline{b}_M)\) is non-negative. It suffices to show the fraction is also non-negative. Substituting \(q \ge \overline{b}_M\), we get
\begin{align*}
    \frac{\zeta B (1-\zeta) - \zeta^2 k \overline{b}_M + \zeta k q}{k+1} \ge \frac{\zeta B (1-\zeta) - \zeta^2 k \overline{b}_M + \zeta k \overline{b}_M}{k+1} = \frac{\zeta (1-\zeta) (B + k\overline{b}_M)}{k+1} \ge 0.
\end{align*}
Therefore, bidding \(B\) is a weakly dominant strategy for sellers with \(q \in [\bar{b}_M, B)\). 

    To show that bidding the quality \(q\) is a weakly dominant strategy for a seller with quality \(q \in (\underline{b}_i, \bar{b}_i)\) for some \(i = 1, \ldots, M\), it suffices to compare bidding \(q\) with bidding \(B\), as all other possible bids are covered by the proof of \autoref{l:wd}. If she bids her quality \(q\), as described in the proof of \autoref{l:wd}, her payoff is at least zero when \(b^{\min}_{-} \le \bar{b}_M\). When \(b^{\min}_{-} = B\) and \(k\) other sellers bid \(B\), with probability \(1-\zeta\) all other sellers are disqualified (i.e., no other sellers submit a valid bid), and therefore this seller gets \(b^{\min}_{-} - q = B - q\); with probability \(\zeta\), this seller's payoff is given by \(\zeta\left(\frac{B+k\bar{b}_M}{k+1} - q\right)\). If she bids \(B\) instead, with probability \(1-\zeta\) she is disqualified and gets zero payoff, and with probability \(\zeta\) her payoff is \(\zeta(B-q)/(k+1)\), which is strictly worse than bidding \(q\) either way. This shows that bidding the quality \(q\) is a weakly dominant strategy for a seller with quality \(q \in (\underline{b}_i, \bar{b}_i)\) for some \(i = 1, \ldots, M\). An analogous argument shows that for a seller with quality \(q \in [\bar{b}_{i-1}, \underline{b}_i]\) for some \(i = 1, \ldots, M\), a weakly dominant strategy is bidding \(\underline{b}_i\). 
\end{proof}

\autoref{t:bra_wo} establishes the optimality of aBRA.

\begin{theorem} \label{t:bra_wo}
    Let \(\lambda^*\) be the Lagrangian multiplier associated with the optimal interim allocation identified in \autoref{p:wia}. If \(\phi_{\gamma}(q;\lambda^*) = h_{\gamma}(q) + \lambda^* g(q)\) is structured, then
    \begin{enumerate}
        \item[(a)] if \(S_+ = S_0\), the BRA described in \autoref{t:bra_bo} with \(\overline{S}\) replaced by \(S_0\) is optimal.
        \item[(b)] if \(S_+ < S_0\), an aBRA with standard bid intervals as described in \autoref{t:bra_bo} with \(\overline{S}\) replaced by \(S_+\), extra bid \(F^{-1}(S_0)\), and qualification rate \(\zeta = \overline{P}/A(S_+, S_0)\) is optimal. 
    \end{enumerate}
\end{theorem}

\begin{proof}
Part (a) follows from \autoref{t:bra_bo}, and it suffices to prove Part (b). Using \autoref{l:wd_abra}, the proof of \autoref{t:bra_bo} indicates that the interim allocation probability induced by the aBRA, \(P_{aBRA}\), equals \(\hat{P}\) defined by \eqref{eq:wip} in \autoref{p:wia} \(F\)-almost everywhere on \([0,1] \setminus [S_+, S_0)\). For a seller with quality quantile \(s \in [S_+, S_0)\), \autoref{l:wd_abra} implies that it is a weakly dominant strategy for her to bid \(B\). With probability \(1-\zeta\), this seller is disqualified, which means that she wins with probability \(0\). With probability \(\zeta\), this seller is qualified; following the same steps as in the proof of \autoref{t:bra_bo}, we see that she wins the auction with probability \(\int_{S_+}^{S_0} (1-s)^{n-1} \, \mathrm{d}s/(S_0 - S_+)\). Thus, the interim allocation probability for this seller is
\[P_{aBRA}(s) = \zeta \cdot \frac{\int_{S_+}^{S_0} (1-s)^{n-1} \, \mathrm{d}s}{S_0 - S_+} + (1-\zeta) \cdot 0 = \overline{P},\]
where the second equality follows from the definition of the qualification rate \(\zeta\). Therefore, \(P_{aBRA} = \hat{P}\) \(F\)-almost everywhere, which implies that the aBRA described in Part (b) of the theorem is indeed optimal. \end{proof}

\autoref{cor:spa} and \autoref{cor:single_peaked} identify environments in which some simple trading mechanisms are optimal. 

\begin{corollary} \label{cor:spa}
    Suppose \(v(q) - q\) is strictly decreasing, \(F\) is twice continuously differentiable, and both \(F\) and \(1-F\) are log-concave. Then a second-price auction with reserve price \(F^{-1}(S_0)\) is optimal. 
\end{corollary}

\begin{proof}
    We first claim that \(-F/f\) is decreasing. To see this, differentiate to obtain \((-F/f)' = -1-(-Ff'/f^2)\). When \(f'(q) \le 0\), because \(1-F\) is log-concave, \(-Ff'/f^2 \ge (1-F)f'/f^2 \ge -1\); if instead \(f'(q) > 0\), \(-Ff'/f^2 \ge -1\) since \(F\) is log-concave. Thus, \((-F/f)' \le 0\). Then since \(v(q) - q\) is strictly decreasing, for any \(\gamma \in [0,1]\), both \(h_{\gamma}(q)\) and \(g(q)\) are strictly decreasing, and so are \(\widetilde{g}\) and \(\widetilde{h}_{\gamma}\). Therefore, the virtual surplus of the Lagrangian, \(\widetilde{\phi}_{\gamma}(s;\lambda) = \widetilde{h}_{\gamma}(s) + \lambda \widetilde{g}(s)\) must be strictly decreasing for any \(\lambda \ge 0\) and \(\gamma \in [0,1]\). Consequently, for any \(\lambda \ge 0\), there must exist a unique \(S_0 \in [0,1]\) such that \(\widetilde{\phi}(S_0;\lambda) = 0\). The result thus follows.
\end{proof}

\autoref{cor:spa} specifies that if the buyer values marginal quality uniformly less than the potential sellers do, under standard distributional assumptions in the mechanism design literature, a second-price auction (with a reserve price if needed) maximizes any weighted average of the buyer's expected payoff and the social surplus. 

\begin{corollary} \label{cor:single_peaked}
	If \(v(q)\) is concave and \(F/f\) is convex in \(q\), then \(\phi_{\gamma}(q;\lambda^*)\) is single-peaked for all \(\lambda^* \ge 0\). Consequently, 
		\begin{itemize}
			\item if there do not exist \(0 \le S_+ < S_0 \le 1\) such that \(\phi_{\gamma}(\cdot;\lambda^*) = 0\) on \([S_+,S_0]\), the BRA with one bid interval is optimal;
			\item otherwise, an aBRA with one standard bid interval \([\underline{b}, F^{-1}(S_+)]\) and extra bid \(F^{-1}(S_0)\) is optimal. 
		\end{itemize}
\end{corollary}

\begin{proof}
If \(\overline{\phi}_{\gamma}(0; \lambda^*) < 0\), then it is optimal not to buy from any potential seller. This is equivalent to a (trivial) BRA with the lone bid interval \(\{0\}\).

Now suppose \(\overline{\phi}_{\gamma}(0; \lambda^*) = m \ge 0\). Because both \(v(q)-q\) and \(-F(q)/f(q)\) are concave under our assumptions, so is \(\phi_{\gamma}(q;\lambda^*)\) since it is a non-negative linear combination of \(v(q)-q\) and \(-F(q)/f(q)\). Thus, \(\phi_{\gamma}(q;\lambda)\) is single-peaked. Consequently, there exists \(c \in [0,1]\) such that \(\overline{\phi}_{\gamma}(\, \cdot \,; \lambda^*) = m\) on \([0,c]\) and decreasing on \([c,1]\). Then by \autoref{t:bra_wo}, if \(S_+ = S_0\), a BRA with one bid interval \([\underline{b}, F^{-1}(S_0)]\) is optimal; otherwise, an aBRA with one standard bid interval \([\underline{b}, F^{-1}(S_+)]\) and extra bid \(F^{-1}(S_0)\) is optimal. 
\end{proof}

Assuming that \(v\) is concave implies that the buyer's marginal valuation of quality is decreasing. As pointed out by \cite{lpv}, the convexity of \(F/f\) is satisfied by many familiar distributions with bounded supports, including power distributions,\footnote{The CDF of a power distribution takes the form of \(F(x) = x^\alpha\), where \(\alpha > 0\). When \(\alpha = 1\), we get the uniform distribution.} (truncated) Pareto distributions, and (truncated) exponential distributions. In fact, it is also satisfied by Beta distributions with both parameters greater than or equal to 1.

\section{An extension of the main model} \label{dnp}
In the main model, the buyer's valuation is assumed to be a deterministic function of the sellers' quality. Moreover, a seller's cost, or her reservation value, is identical to her quality. These assumptions are made to simplify notation and emphasize the quality concerns. In many relevant applications, however, it might be more natural either to assume that the buyer's valuation is a random variable or to treat the sellers' costs as primitives, or both. In what follows, we show that, even in these settings, focusing on the main model is without loss.\footnote{For brevity, we only discuss the buyer-optimal problem here; extending the analysis to allow for any weighted average is straightforward.}

For concreteness, consider a buyer who would like to contract with one of several potential suppliers to develop a new project, say a new production line. The cost of supplier \(s\), \(c_s \in [0,1]\), is her private information; the costs are independently and identically distributed according to a continuous density function \(f_C(\cdot)\). The project's value is not perfectly revealed to the buyer until the end of the development phase at the earliest, which is long after penning the contract. Consequently, at the time of contracting the buyer's valuation is a random variable \(\Xi\). We assume that the realization of \(\Xi\), \(\xi \in [\underline{\xi},\overline{\xi}]\), is not contractible. The buyer believes that \(\Xi\) and \(C\) are correlated, and that the conditional distribution of \(\Xi\) is \(f_{\Xi \,| C}(\cdot|c)\). One possible assumption can be that \(\Xi\) and \(C\) are positively affiliated, or equivalently \(\text{MTP}_2\) (\citeauthor{kr80a}, \citeyear{kr80a}; see also \citeauthor{mw82}, \citeyear{mw82}).

Let \(\boldsymbol{c} \coloneqq (c_1, \ldots, c_n)\). Given a direct mechanism \(\{p_s(\boldsymbol{c}), t_s(\boldsymbol{c})\}_{s=1}^n\), where for each cost profile \(\boldsymbol{c}\), \(p_s(\boldsymbol{c})\) specifies the probability that the buyer contracts with supplier \(s\), and \(t_s(\boldsymbol{c})\) is the transfer that the buyer pays to supplier \(s\), the buyer's expected payoff can be written as
\begin{align*}
	\pi_{b} & = \int_{[0,1]^{n}} \left[\sum_{s=1}^{n}\int_{\underline{\xi}}^{\overline{\xi}} (\xi \,p_{s}(\boldsymbol{c})-t_{s}(\boldsymbol{c})) f_{\Xi \,| C}(\xi|c_s) \, \mathrm{d}\xi \right] f^n(\boldsymbol{c}) \, \mathrm{d}\boldsymbol{c} \\
		& = \int_{[0,1]^{n}} \sum_{s=1}^{n} \left[\left(\int_{\underline{\xi}}^{\overline{\xi}} \xi f_{\Xi \,| C}(\xi|c_s) \, \mathrm{d}\xi\right) p_{s}(\boldsymbol{c})-t_{s}(\boldsymbol{c}) \right] f^n(\boldsymbol{c}) \, \mathrm{d}\boldsymbol{c} \\
		& = \int_{[0,1]^{n}} \sum_{s=1}^{n}\left(\mathbb{E}[\Xi \, | \, C=c_s]\, p_{s}(\boldsymbol{c})-t_{s}(\boldsymbol{c}) \right) f^n(\boldsymbol{c}) \, \mathrm{d}\boldsymbol{c},
\end{align*}
where \(f^n(\boldsymbol{c}) := \prod_{s=1}^n f_C(c_s)\). If we define \(v(c_s):= \mathbb{E}[\Xi \, | \, C=c_s]\), we see from \eqref{bbp} that the problem here is identical to the procurement problem we study in the main text, and the curvature of \(v(c_s)\) is governed by the conditional distribution. For example, if \(\Xi\) and \(C\) are positively affiliated, \(v(\cdot)\) is increasing.

\begin{example} \label{example:pareto}
	A manufacturer would like to procure a machine for production. For simplicity, suppose that her valuation is identical to the durability of the machine. She believes that the two potential sellers' costs are identically, independently, and uniformly distributed on \([0,1]\). Conditional on the cost realization \(c\), her valuation \(\Xi\) is distributed according to a Pareto distribution with scale 0.5 and shape \(2.2-c\).\footnote{The scale parameter can be interpreted as the length of the machine's warranty.} Consequently, 
	\[
	v(c) = \mathbb{E}[\Xi \, | \, C=c] = \frac{1.1-0.5c}{1.2-c},
	\]
	and \(g(c) = v(c) - 2c\). 
    By \autoref{p:oia}, the optimal interim allocation is given by
    \[
    \hat{P}(c) = 
    \left\{\begin{array}{ll}
    	1-c & c < 0.48, \\
    	0.26 & c \ge 0.48.
    \end{array}\right.
    \]
    By \autoref{cor:bra_sc}, a BRA with two bid intervals \([0,0.48]\) and \(\{1\}\) is optimal. \hfill \(\lozenge\)
\end{example}
\end{appendices}

\end{document}